%% file: ms.tex
\documentclass[acmsmall, authorversion]{acmart}

\usepackage{prelude}

\usemathitv

\theoremstyle{acmcase}
\newtheorem{step}{Step}
\cCrefname{step}{Step}{Steps}

\newcommand{\iftwocol}{\ifdefstring{\acmformat}{sigconf}}
\makeatletter
\newcommand{\ifanonymous}[2]{\if@ACM@anonymous#1\else#2\fi}
\makeatother

\ifanonymous{\includecomment{anononly}}{\excludecomment{anononly}}

\newlength{\captionsqueeze}
\iftwocol{\setlength{\captionsqueeze}{\baselineskip}}{}

\newcommand{\proofin}[2][]{%
  \IfSubStr{#1}{,}{\begin{proof}[Proofs of \cref{#1}]}{\begin{proof}}
    See \cref{#2}. \noqed
  \end{proof}}

\newcommand{\tail}[2][]{\appopt{\subopt{\ol{F}}{#1}}{#2}}
\newcommand{\density}[2][]{\appopt{\subopt{f}{#1}}{#2}}
\newcommand{\hazard}[2][]{\appopt{\subopt{h}{#1}}{#2}}
\newcommand{\gail}[2][]{\appopt{\subopt{\ol{G}}{#1}}{#2}}
\newcommand{\gensity}[2][]{\appopt{\subopt{g}{#1}}{#2}}

\newsavebox{\upbox}
\sbox{\upbox}{%
  \begin{tikzpicture}[baseline=0cm, transform canvas={scale=0.6}]
    \draw[line width=0.4 / 0.6, ->] (0, 0) -- (0, 0.14 / 0.6);
  \end{tikzpicture}}
\newcommand{\up}{{\mkern 3mu\vphantom{+}\usebox{\upbox}\mkern 3mu}}

\newsavebox{\mathquotequotes}
\newsavebox{\mathquotecontent}
\makeatletter
\newcommand{\mathquote}[1]{%
  \savebox{\mathquotequotes}{\text{``}}%
  \savebox{\mathquotecontent}{$\displaystyle #1$}%
  \raisebox{\dimexpr\ht\mathquotecontent-\ht\mathquotequotes\relax}{``}%
  #1%
  \raisebox{\dimexpr\ht\mathquotecontent-\ht\mathquotequotes\relax}{''}}

\newcommand{\rank}[2][]{\appopt{\subopt{r}{#1}}{#2}}
\newcommand{\rankp}[2][]{\appopt{\subopt{r'}{#1}}{#2}}
\newcommand{\rankup}[2][]{\appopt{\subopt{r^\up}{#1}}{#2}}

\newcommand{\y}{y}
\newcommand{\z}{z}

\newcommand{\coload}[1]{\appopt{\ol{\rho}}{#1}}
\newcommand{\excess}[1]{\appopt{\tau}{#1}}
\newcommand{\waiting}[2][]{\appopt{\subopt{Q}{#1}}{#2}}
\newcommand{\residence}[2][]{\appopt{\subopt{R}{#1}}{#2}}
\newcommand{\response}[2][]{\appopt{\subopt{T}{#1}}{#2}}

\newcommand{\policyname}{\textsf}
\newcommand{\gittins}{\policyname{Gittins}}
\newcommand{\serpt}{\policyname{SERPT}}
\newcommand{\mserptmath}{\policyname{M\=/SERPT}}
\newcommand{\mserpt}{\ifmmode\mserptmath\else M\=/SERPT\fi}

\newcommand{\shortensubscripts}{%
  \renewcommand{\gittins}{\policyname{G}}
  \renewcommand{\serpt}{\policyname{S}}
  \renewcommand{\mserptmath}{\policyname{MS}}}

\makeatletter
\catcode`\$=3 

\usetikzlibrary{backgrounds, math, shapes, positioning, decorations.markings}

\def\@above{above}
\def\@below{below}
\def\@negative{-}
\newcommand{\downmode}{%
  \def\@above{below}
  \def\@below{above}
  \def\@negative{}}

\newcommand{\xscale}{10}
\newcommand{\yscale}{6}
\newcommand{\ticksize}{3.7}
\newcommand{\arrowsize}{8}
\newcommand{\xarrowsize}{(\arrowsize/\xscale)}
\newcommand{\yarrowsize}{(\arrowsize/\yscale)}

\newcommand{\xscaleload}{%
  \renewcommand{\xscale}{150}
  \renewcommand{\yscale}{10}}

\newcommand{\projx}[1]{($ (0,0)!#1!(1,0) $)}
\newcommand{\projxheight}[2]{($ (0,#2)!#1!(1,#2) $)}
\newcommand{\midpt}[2]{($ #1!0.5!#2 $)}

\newcommand{\xguide}[3][]{%
  \draw[axis] ({#2}, 0) -- ({#2}, {\@negative(\ticksize)/\yscale})
  node[\@below] {\ifempty{#1}{$#2$}{#1}\vphantom{by}};
  \draw[guide] ({#2}, 0) -- ({#2}, {#3});}
\newcommand{\yguide}[3][]{%
  \draw[axis] (0, {#3}) -- (-\ticksize/\xscale, {#3})
  node[left] {\ifempty{#1}{$#3$}{#1}};
  \draw[guide] (0, {#3}) -- ({#2}, {#3});}

\newcommand{\xguidept}[2]{%
  \ifempty{#1}{}{%
    \draw[axis] \projx{#2} -- \projxheight{#2}{{\@negative(\ticksize)/\yscale}}
    node[\@below] {#1\vphantom{by}}};
  \draw[guide] \projx{#2} -- #2;}

\newcommand{\axes}[4]{%
  \draw[axis, ->] (-\ticksize/\xscale, 0)
  node[left] {$0$\vphantom{by}} -- ({#1 + \xarrowsize}, 0)
  node[right] {#3\vphantom{by}};
  \draw[axis, ->] (0, \@negative\ticksize/\yscale)
  node[\@below] {$0$\vphantom{by}} -- (0, {#2 + \yarrowsize})
  node[\@above] {#4\vphantom{by}};}

\newcommand{\axesxtorsor}[4]{%
  \draw[axis, ->] (-\ticksize/\xscale, 0)
  node[left] {$0$\vphantom{by}} -- ({#1 + \xarrowsize}, 0)
  node[right] {#3\vphantom{by}};
  \draw[axis, ->] (0, 0) -- (0, {#2 + \yarrowsize})
  node[\@above] {#4};}

\newcommand{\snakefirst}[2]{#1 parabola bend #1 \midpt{#1}{#2}}
\newcommand{\snakesecond}[2]{\midpt{#1}{#2} parabola bend #2 #2}
\newcommand{\snake}[2]{\snakefirst{#1}{#2} -- \snakesecond{#1}{#2}}

\newcommand{\intervallabel}[3]{%
  \begin{scope}
    \node (A) at \projx{#1} {};
    \node (B) at \projx{#2} {};
    \node[\@above=1.5pt, inner sep=0pt] at \midpt{(A)}{(B)} {%
      \footnotesize #3\vphantom{ly}};
  \end{scope}}
\newcommand{\intervalsnake}[5][--]{%
  \begin{scope}[on background layer]
    \fill[#2] \projx{#3} -- \snake{#3}{#4} -- \projx{#4};
  \end{scope}
  \draw[guide] \projx{#3} #1 \snake{#3}{#4};
  \intervallabel{#3}{#4}{#5}}
\newcommand{\intervalsnakesecond}[5][--]{%
  \begin{scope}[on background layer]
    \fill[#2] \projx{#3} -- #3 parabola bend #4 #4 -- \projx{#4};
  \end{scope}
  \draw[guide] \projx{#3} #1 #3 parabola bend #4 #4;
  \intervallabel{#3}{#4}{#5}}
\newcommand{\intervalspecial}[6][--]{%
  \begin{scope}[on background layer]
    \fill[#2]
      \projx{#3}
      -- #3
      parabola bend #4 #4
      parabola bend #4 #5
      -- \projx{#5};
  \end{scope}
  \draw[guide] \projx{#3} #1 #3;
  \intervallabel{#3}{#5}{#6}}

\tikzset{%
  figure/.style={x=\xscale, y=\yscale, thick, font=\small,
    baseline={([yshift={-\ht\strutbox}]current bounding box.north)}},
  axis/.style={thick},
  guide/.style={black!42, thick, densely dotted},
  primary/.style={ultra thick, color=cyan!97!black},
  secondary/.style={ultra thick, color=green!41!yellow!78!black},
  cutoff/.style={line width=3.2pt, dotted, color=magenta!92},
  hill/.style={color=red!37!yellow!25},
  valley/.style={color=green!74!yellow!24}}

\catcode`\$=13 
\makeatother

\setcopyright{acmcopyright}
\acmJournal{POMACS}
\acmYear{2020}
\acmVolume{4}
\acmNumber{1}
\acmArticle{11}
\acmMonth{3}
\acmPrice{15.00}
\acmDOI{10.1145/3379477}

\received{October 2019}
\received[revised]{December 2019}
\received[accepted]{January 2020}

\ccsdesc[500]{General and reference~Performance}
\ccsdesc[500]{Mathematics of computing~Queueing theory}
\ccsdesc[500]{Networks~Network performance modeling}
\ccsdesc[300]{Theory of computation~Routing and network design problems}
\ccsdesc[300]{Computing methodologies~Model development and analysis}
\ccsdesc[300]{Software and its engineering~Scheduling}

\keywords{%
  M/G/1;
  response time;
  latency;
  sojourn time;
  Gittins policy;
  shortest expected remaining processing time (SERPT),
  monotonic SERPT (M-SERPT);
  approximation ratio;
  multilevel processor sharing (MLPS);
  foreground-background (FB);
  shortest remaining processing time (SRPT)}

\begin{document}

\title{Simple Near-Optimal Scheduling for the M/G/1}
\ifanonymous{%
  \subtitle{%
    One-shot revision of 2020 summer deadline \#15.
    \color{red}%
    See \cref{app:rebuttal} for responses to reviewer comments.}}{}

\author{Ziv Scully}
\affiliation{%
  \institution{Carnegie Mellon University}
  \department{Computer Science Department}
  \streetaddress{5000 Forbes Ave}
  \city{Pittsburgh}
  \state{PA}
  \postcode{15213}
  \country{USA}}
\email{zscully@cs.cmu.edu}

\author{Mor Harchol-Balter}
\affiliation{%
  \institution{Carnegie Mellon University}
  \department{Computer Science Department}
  \streetaddress{5000 Forbes Ave}
  \city{Pittsburgh}
  \state{PA}
  \postcode{15213}
  \country{USA}}
\email{harchol@cs.cmu.edu}

\author{Alan Scheller-Wolf}
\affiliation{%
  \institution{Carnegie Mellon University}
  \department{Tepper School of Business}
  \streetaddress{5000 Forbes Ave}
  \city{Pittsburgh}
  \state{PA}
  \postcode{15213}
  \country{USA}}
\email{awolf@andrew.cmu.edu}

\begin{abstract}
  We consider the problem of preemptively scheduling jobs
  to minimize mean response time of an M/G/1 queue.
  When we know each job's size,
  the \emph{shortest remaining processing time} (SRPT) policy is optimal.
  Unfortunately,
  in many settings we do not have access to each job's size.
  Instead, we know only the job size distribution.
  In this setting the \emph{Gittins} policy
  is known to minimize mean response time,
  but its complex priority structure can be computationally intractable.
  A much simpler alternative to Gittins is the
  \emph{shortest expected remaining processing time} (SERPT) policy.
  While SERPT is a natural extension of SRPT to unknown job sizes,
  it is unknown whether or not SERPT
  is close to optimal for mean response time.

  We present a new variant of SERPT called \emph{monotonic SERPT} (\mserpt{})
  which is as simple as SERPT
  but has provably near-optimal mean response time at all loads
  for any job size distribution.
  Specifically,
  we prove the mean response time ratio between \mserpt{} and Gittins is
  at most~$3$ for load $\rho \leq 8/9$
  and at most~$5$ for any load.
  This makes \mserpt{} the only non-Gittins scheduling policy
  known to have a constant-factor approximation ratio for mean response time.
\end{abstract}

\maketitle
\newlength{\colwidth}
\iftwocol{%
  \setlength{\colwidth}{\linewidth}}{%
  \setlength{\colwidth}{0.5\linewidth}}

\section{Introduction}

Scheduling to minimize mean response time in a preemptive M/G/1 queue
is a classic problem in queueing theory.
When job sizes are known,
the \emph{shortest remaining processing time} (SRPT) policy
is known to minimize mean response time \citep{srpt_optimal_schrage}.
Unfortunately, determining or estimating a job's exact size
is difficult or impossible in many applications,
in which case SRPT is impossible to implement.
In such cases we only learn jobs' sizes after they have completed,
which can give us a good estimate of the \emph{distribution} of job sizes.

When individual job sizes are unknown but the job size distribution is known,
the \emph{Gittins} policy minimizes mean response time
\citep{m/g/1_gittins_aalto, book_gittins}.
Gittins has a seemingly simple structure:
\begin{itemize}
\item
  Based on the job size distribution,
  Gittins defines a \emph{rank function}
  that maps a job's \emph{age},
  which is the amount of service it has received so far,
  to a \emph{rank}, which denotes its priority \citep{soap_scully}.
\item
  At every moment in time,
  Gittins applies the rank function to each job's age
  and serves the job with the best rank.
\end{itemize}
Unfortunately, hidden in this simple outline is a major obstacle:
computing the rank function from the job size distribution requires
solving a nonconvex optimization problem for every possible age.
Although the optimization can be simplified for
specific classes of job size distributions \citep{m/g/1_gittins_aalto},
it is intractable in general.

In light of the difficulty of computing the Gittins rank function,
practitioners turn to a wide variety of simpler scheduling policies,
each of which has good performance in certain settings.
Three of the most famous are the following:
\begin{itemize}
\item
  \emph{First-come, first-serve} (FCFS)
  serves jobs nonpreemptively in the order they arrive.
  \begin{itemize}
  \item
    FCFS generally performs well for low-variance job size distributions
    and is optimal for those with the
    \emph{new better than used in expectation} property
    \citep{dhr_dmrl_optimality_righter, m/g/1_gittins_aalto}.
  \end{itemize}
\item
  \emph{Foreground-background} (FB)
  always serves the job of minimal age,
  splitting the server evenly in case of ties.
  \begin{itemize}
  \item
    FB generally performs well for high-variance job size distributions
    and is optimal for those with the \emph{decreasing hazard rate} property
    \citep{fb_optimality_misra, dhr_dmrl_optimality_righter,
      dhr_ihr_optimality_righter, m/g/1_gittins_aalto}.
  \end{itemize}
\item
  \emph{Processor sharing} (PS)
  splits the server evenly between all jobs currently in the system.
  \begin{itemize}
  \item
    PS has appealing insensitivity
    \citep{ps_analysis_kleinrock, ps_insensitive_bonald,
      ps_insensitive_cheung}
    and fairness \citep{fairness_raz, fairness_wierman} properties
    which ensure passable mean response time for all job size distributions,
    but it is only optimal in the trivial special case of
    exponential job size distributions.
  \end{itemize}
\end{itemize}
These are a few of the many scheduling heuristics
studied in the past several decades
\citep{book_harchol-balter, book_kleinrock, lps_tail_zwart,
  lps_design_yamazaki, mlps_analysis_guo, mlps_delay_aalto,
  mlps_two-level_aalto, smart_insensitive_wierman}.
Unfortunately, there are
\emph{no guarantees of near-optimal mean response time}
for any non-Gittins policy
that hold across all job size distributions.
In fact, we show in \cref{app:infinite_ratio} that
FCFS, FB, and PS
can have \emph{infinite} mean response time ratio compared to Gittins.
We therefore ask:
\begin{quote}
  Is there a \emph{simple} scheduling policy
  with near-optimal mean response time
  for all job size distributions?
\end{quote}

\newcommand{\figresponsecomparison}{%
  \centering
  \begin{minipage}[t]{\colwidth}
    \centering
    \textsc{Mean Response Time Relative to Gittins} \\
    \includegraphics[width=\linewidth]{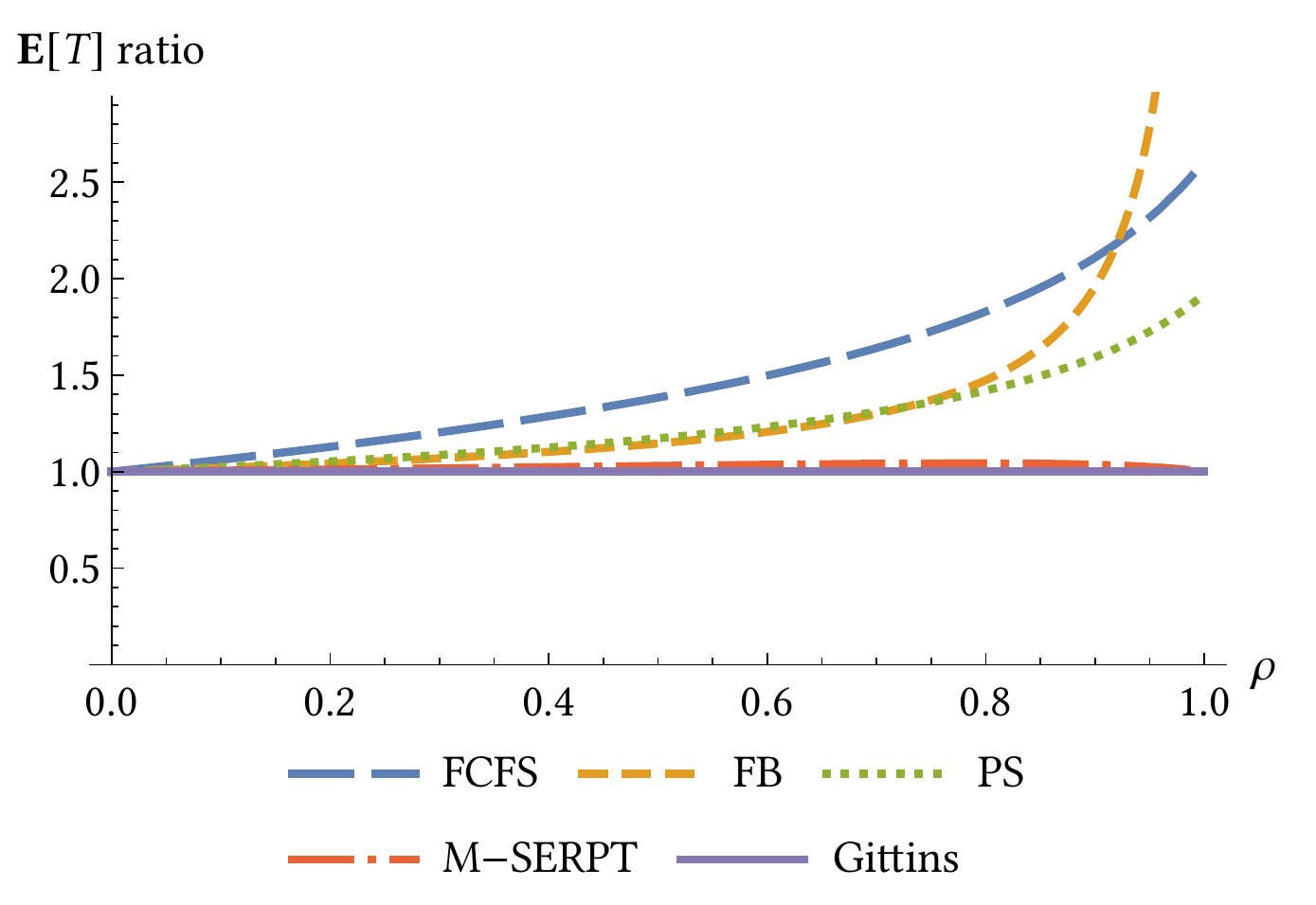}
  \end{minipage}%
  \hfill
  \begin{minipage}[t]{\colwidth}
    \centering
    \textsc{Job Size Distribution} \\
    \includegraphics[width=\linewidth]{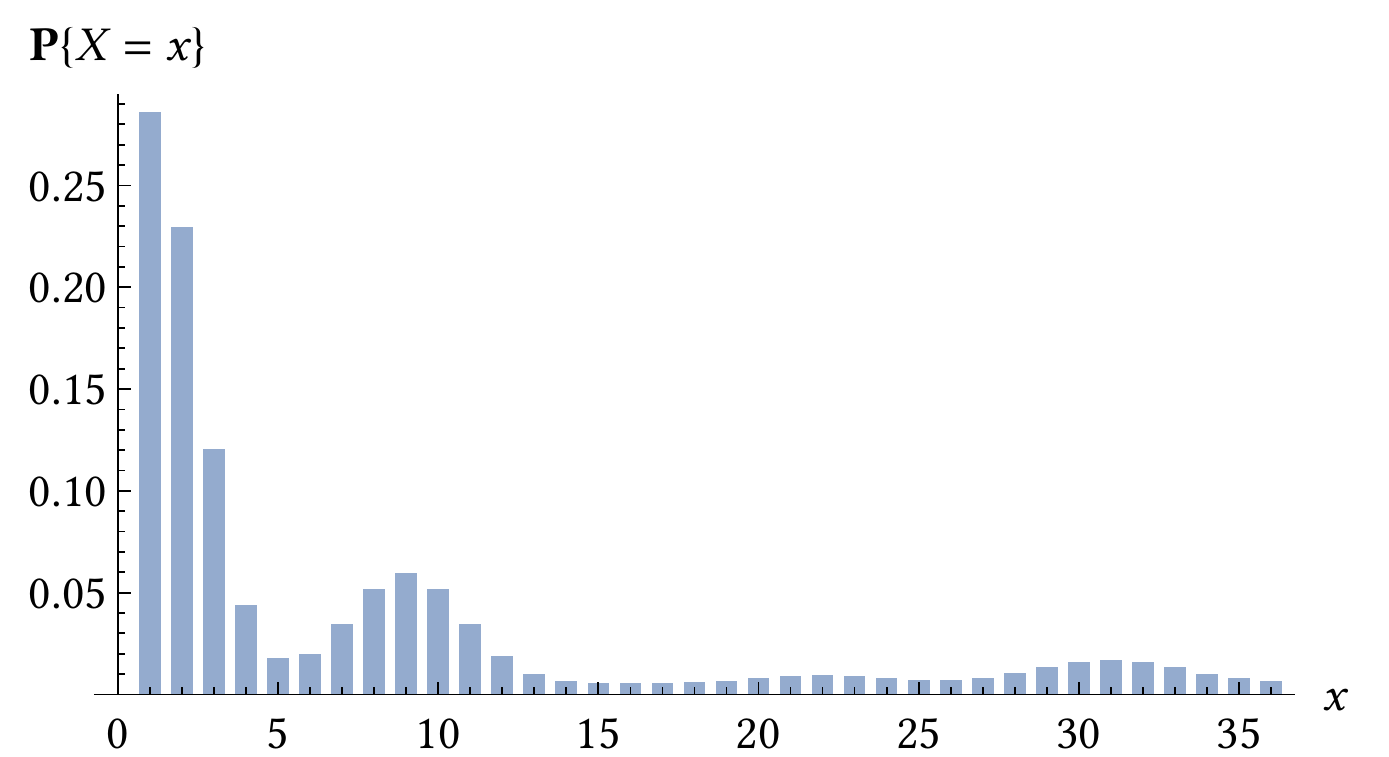}
  \end{minipage}
  \vspace{-\captionsqueeze}
  \caption{Mean Response Time Comparison}
  \label{fig:response_comparison}}

\newcommand{\tabcomplexity}{%
  \begin{table}
    \centering
    \caption{Comparison of Gittins, SERPT, and \mserpt{}}
    \label{tab:complexity}
    \vspace{-0.5\captionsqueeze}
    \begin{tabular}{@{}llll@{}}
      \toprule
      \textsc{Policy}
      & \multicolumn{2}{c}{\textsc{Computation}}
      & \textsc{Optimality} \\
      \cmidrule(lr){2-3}
      & \textit{Discrete} & \textit{Continuous} & \\
      \midrule
      Gittins & $O(n^2)$ & intractable & optimal \\
      SERPT & $O(n)$ & tractable & unknown \\
      \mserpt & $O(n)$ & tractable & $5$\=/approximation or better \\
      \bottomrule
    \end{tabular}
  \end{table}}

\iftwocol{\begin{figure*}\figresponsecomparison\end{figure*}}{}

\iftwocol{\tabcomplexity}{}

One candidate for such a policy is
\emph{shortest expected remaining processing time} (SERPT).
Like Gittins, SERPT assigns each job a rank as a function of its age,
but SERPT has a much simpler rank function:
a job's rank is its \emph{expected remaining size}.
That is, if the job size distribution is~$X$,
then under SERPT, a job's rank at age~$a$ is
\begin{align*}
  \rank[\serpt]{a} = \E{X - a \given X > a},
\end{align*}
where lower rank means better priority.
Intuitively, it seems like SERPT should have low mean response time
because it prioritizes jobs that are short in expectation,
analogous to what SRPT does for known job sizes.
SERPT is certainly much simpler than Gittins,
as summarized in \cref{tab:complexity}
and discussed in detail in \cref{app:gittins_hard}.
\begin{itemize}
\item
  For discrete job size distributions with $n$ support points,
  the best known algorithms compute Gittins's rank function
  in $O(n^2)$ time \citep{gittins_index_computation_chakravorty}.
  In contrast, SERPT's rank function takes just $O(n)$ time to compute.
\item
  For continuous job size distributions,
  computing Gittins's rank function
  is intractable with known methods:
  it requires solving a nonconvex optimization problem
  at every age~$a$,
  and the objective of the optimization
  requires numerical integration to compute.
  In contrast, SERPT's rank function requires just numerical integration.
\end{itemize}

\iftwocol{}{\tabcomplexity}

\subsection{Challenges}
\label{sub:challenges}

SERPT is intuitively appealing and simple to compute,
but does it have near-optimal mean response time?
This question is open:
there is \emph{no known bound} on the performance gap
between SERPT and Gittins.
To be precise, letting\footnote{%
  The mean response time ratio $C_\serpt(X)$
  also depends on the load~$\rho$,
  but we omit~$\rho$ from the notation to reduce clutter.}
\begin{align*}
  C_\serpt(X)
  = \frac{\E{\response[\serpt]{X}}}{\E{\response[\gittins]{X}}}
\end{align*}
be the mean response time ratio between SERPT and Gittins
for a given job size distribution~$X$,
there is no known bound on
\begin{align*}
  \text{approximation ratio of SERPT} = \sup_X C_\serpt(X).
\end{align*}
This approximation ratio is difficult to bound because
we have to consider \emph{all} possible job size distributions~$X$.

In fact, until recently it was unknown how to compute $C_\serpt(X)$
even given a \emph{specific} job size distribution~$X$.
This changed with the introduction of the \emph{SOAP} technique
\citep{soap_scully},
which can analyze the mean response time of any scheduling policy
that can be specified by a rank function.
We can use SOAP to \emph{numerically} compute $C_\serpt(X)$
for any given job size distribution~$X$.
However, SOAP does not give a bound on SERPT's approximation ratio,
which requires considering all possible~$X$.

One might hope to derive a general expression for $C_\serpt(X)$ using SOAP.
While this is possible in principle,
the resulting expression is intractable (\cref{sub:why_not_soap}).
In light of this, our strategy is to create a new scheduling policy
that captures the essence of SERPT
but has a tractable mean response time expression
in terms of~$X$.

\subsection{A New Simple Scheduling Policy: \mserpt{}}

In this paper we introduce a new policy
called \emph{monotonic SERPT} (\mserpt{})
that is simple to compute
and has provably near-optimal mean response time.
Like Gittins and SERPT, we specify \mserpt{} using a rank function.
\mserpt{}'s rank function is like SERPT's,
except \emph{a job's rank never improves:}
\begin{align*}
  \rank[\mserpt]{a}
  = \max_{0 \leq b \leq a} \rank[\serpt]{b}.
\end{align*}

We prove that \mserpt{} is a $5$\=/approximation for mean response time,
meaning its mean response time is at most $5$~times that of Gittins.
This makes \mserpt{} the first non-Gittins scheduling policy
known to have a constant-factor approximation ratio.
The approximation ratio is even smaller at low and moderate loads.
For example, \mserpt{} is a $3$\=/approximation for load $\rho \leq 8/9$.
Remarkably, \mserpt{} achieves its constant-factor approximation ratio
with a rank function that is as simple to compute as SERPT's
(\cref{tab:complexity}).

\iftwocol{}{\begin{figure}\figresponsecomparison\end{figure}}

Our approximation ratio for \mserpt{} is a worst-case upper bound.
There are many distributions where \mserpt{}'s performance is
equal or very close to Gittins's.
For example, \cref{fig:response_comparison} compares the mean response times
of several policies, including \mserpt{}, to that of Gittins,
where the job size distribution is the mixture of four bell curves pictured.
In this example, \mserpt{}'s mean response time is within $4\%$ of Gittins's
across all loads.\footnote{%
  For the specific distribution in \cref{fig:response_comparison},
  SERPT has mean response time between Gittins and \mserpt{}.
  However, there are examples where SERPT has
  greater mean response time than \mserpt{},
  and whether SERPT has a constant-factor approximation ratio
  remains an open problem.}
In further preliminary numerical experiments,
omitted for lack of space,
we only observed a mean response time difference of more than $15\%$
in a specific pathological scenario (\cref{sec:pathological}).

\subsection{Contributions}

We introduce \mserpt{},
the first non-Gittins policy
proven to achieve mean response time within a constant factor of Gittins's.
Our specific contributions are as follows:
\begin{itemize}
\item
  We define the \emph{monotonic SERPT} (\mserpt{}) policy,
  a new variant of SERPT
  (\cref{sec:model}).
\item
  We introduce a new simplification of the SOAP response time analysis
  that yields a tractable mean response time expression for \mserpt{}
  (\cref{sec:outline, sec:hills_valleys}).
\item
  We prove that \mserpt{} is a $5$\=/approximation
  for minimizing mean response time,
  with an even smaller approximation ratio at low and moderate loads
  (\cref{sec:approximation_ratio}).
\item
  We use the fact that \mserpt{} is a $5$\=/approximation
  to resolve two open questions in M/G/1 scheduling theory
  (\cref{sec:implications}).
\item
  We construct a pathological job size distribution for which
  the mean response time ratio between \mserpt{} and Gittins is~$2$,
  which is the largest ratio we have observed
  (\cref{sec:pathological}).
\end{itemize}
\mserpt{}'s approximation ratio is therefore between $2$ and~$5$.
We conclude by discussing in detail why this gap is hard to close
and pointing out several possible avenues of attack
(\cref{sec:tightening}).

\subsection{Related Work}
\label{sub:related_work}

In this paper we consider minimizing mean response time
in the setting of an M/G/1 queue
with unknown job sizes but known job size distribution.
We are not aware of prior work on approximation ratios in this exact setting,
but there is prior work in related settings.

\Citet{smart_insensitive_wierman}
study the M/G/1 with \emph{known} job sizes.
They prove that all scheduling policies in a class called \emph{SMART}
are $2$\=/approximations for mean response time,
where the baseline for this setting is SRPT \citep{srpt_optimal_schrage}.
All SMART policies use job size information,
so they cannot be applied to our setting of unknown job sizes.
Proving approximation ratios in our setting
is significantly more challenging because
the scheduling policies involved, namely \mserpt{} and Gittins,
have much more complicated mean response time formulas
than SRPT and the SMART class \citep{smart_insensitive_wierman, soap_scully}.

We now turn to settings with unknown job sizes.
\Citet{rmlf_pruhs} propose a policy called
\emph{randomized multilevel feedback} (RMLF) for the case
where neither job sizes \emph{nor the job size distribution} are known.
RMLF has been studied in two specific settings:
\begin{itemize}
\item
  In the \emph{worst-case} setting,
  meaning job sizes and arrival times are chosen adversarially,
  RMLF has mean response time $O(\log n)$ times that of SRPT,
  where $n$ is the number of jobs in the arrival sequence
  \citep{rmlf_pruhs, rmlf_leonardi}.
  Up to constant factors,
  this is the best possible performance in the worst-case setting
  \citep{nonclairvoyant_motwani}.
\item
  In the \emph{stochastic} GI/GI/1 setting,
  \citet{rmlf_zwart} prove that as the load $\rho$ approaches~$1$,
  \begin{align*}
    \frac{\E{T_{\policyname{RMLF}}}}{\E{T_{\policyname{SRPT}}}}
    = O\gp[\bigg]{\log\frac{1}{1 - \rho}}.
  \end{align*}
\end{itemize}
These results differ from ours in two important ways.
First, the results do not prove constant-factor approximation ratios:
they give asymptotic ratios that become arbitrarily large
in the $n \to \infty$ and $\rho \to 1$ limits, respectively.
In contrast,
we show that \mserpt{} is a $5$\=/approximation at all loads~$\rho$,
even in the $\rho \to 1$ limit.
Second, the results compare RMLF with SRPT, not with Gittins,
even though job sizes are unknown.
This is because optimal policies for the worst-case and GI/GI/1 settings
are not known,
especially with unknown job size distribution,
leaving SRPT as a sensible baseline for comparison.
In contrast, in the M/G/1 setting with known job size distribution,
we know the optimal policy is Gittins,
so we compare \mserpt{} to Gittins.
Comparing RMLF to Gittins is an interesting open problem.

A final setting is a hybrid between the worst-case and M/G/1 settings.
\Citet{multiple_processors_megow} consider scheduling
jobs with stochastic sizes but adversarially chosen arrival times.
However, rather than considering the metric of mean response time,
they consider mean \emph{completion} time.
The difference between these metrics is that
a job's response time is measured relative to its arrival,
whereas a job's completion time is measured relative to time~$0$.
Completion and response times are only the same
when all the jobs arrive at once.
Thus, while \citet{multiple_processors_megow} show that
Gittins and a related policy are $2$\=/approximations for mean completion time,
this does not translate into an approximation ratio for mean response time.

\section{System Model and Preliminaries}
\label{sec:model}

\shortensubscripts

We consider scheduling policies for a single-class M/G/1 queue
in which jobs have unknown size.
We write $\lambda$ for the arrival rate and $X$ for the job size distribution,
so the load is $\rho = \lambda\E{X}$.
We assume $\rho < 1$ for stability.
Jobs may be preempted at any time without delay or loss of work.

Throughout this paper, all monotonicities are meant in the weak sense
unless otherwise specified.
For example, ``increasing'' means ``nondecreasing''.
Many quantities defined in this paper
depend on one or both of $X$ and~$\rho$,
but we usually leave this implicit in our notation to reduce clutter.

We write $\tail{}$ and $\density{}$ for
the tail and density functions of~$X$, respectively.
For ease of presentation, we assume that
\begin{itemize}
\item
  $\density{}$ is well defined and continuous,
  implying the distribution does not have atoms; and
\item
  both the SERPT rank function (\cref{def:serpt}) and the hazard rate function
  \begin{align*}
    \hazard{x} = \frac{\density{x}}{\tail{x}}
  \end{align*}
  are piecewise monotonic,
  ruling out some pathological cases.
\end{itemize}
With some effort,
one very likely can adapt our proofs to relax these assumptions.
In particular, we have confirmed our results
for discrete job size distributions,
omitting the details for lack of space.

We write $T_\pi$ for the response time distribution under policy~$\pi$,
and we write $T_\pi(x)$ for the response time distribution of
a job of size~$x$ under policy~$\pi$.
We use similar notation for
waiting time~$\waiting[\pi]{}$ and residence time~$\residence[\pi]{}$
(\cref{sub:waiting_residence})
For the most part, $\pi$ is one of
\begin{itemize}
\item
  $\gittins$, denoting Gittins;
\item
  $\serpt$, denoting SERPT; or
\item
  $\mserpt$, denoting \mserpt{}.
\end{itemize}
These policies are defined in \cref{sub:policies}.
We use the same subscripts for other quantities
that depend on the scheduling policy.
We omit the subscript when discussing a generic SOAP policy.

\subsection{SOAP Policies and Rank Functions}
\label{sub:policies}

A SOAP policy \citep{soap_scully} is specified by a \emph{rank function}
\begin{align*}
  \rank{} : \R_{\geq 0} \to \R
\end{align*}
which maps a job's \emph{age}, the amount of time it has been served,
to its \emph{rank}, or priority.\footnote{%
  The full SOAP definition \citep{soap_scully}
  allows a job's rank to also depend on characteristics
  such as its size or class,
  but we do not need this generality for the policies in this paper.}
All SOAP policies have the same core scheduling rule:
always serve the job of \emph{minimum rank},
breaking ties in first-come, first served (FCFS) order.

Gittins, SERPT, and \mserpt{} are all SOAP policies.
Their rank functions are defined as follows.

\begin{definition}
  \label{def:gittins}
  The \emph{Gittins} policy is the SOAP policy with rank function
  \begin{align*}
    \rank[\gittins]{a}
    = \inf_{b > a}
      \frac{\int_a^b \tail{t} \d{t}}{\tail{a} - \tail{b}}.
  \end{align*}
\end{definition}

\begin{definition}
  \label{def:serpt}
  The \emph{shortest expected remaining processing time} (SERPT) policy
  is the SOAP policy with rank function
  \begin{align*}
    \rank[\serpt]{a}
    = \E{X - a \given X > a}
    = \frac{\int_a^\infty \tail{t} \d{t}}{\tail{a}}.
  \end{align*}
\end{definition}

\begin{definition}
  \label{def:increasing_envelope}
  The \emph{increasing envelope} of function~$\rank{}$ is
  \begin{align*}
    \rankup{a}
    = \max_{0 \leq b \leq a} \rank{b}.
  \end{align*}
\end{definition}

\begin{definition}
  \label{def:mserpt}
  The \emph{monotonic SERPT} (\mserpt{}) policy
  is the SOAP policy whose rank function is the increasing envelope
  of SERPT's rank function:
  \begin{align*}
    \rank[\mserpt]{a}
    = \rankup[\serpt]{a}
    = \max_{0 \leq b \leq a} \rank[\serpt]{b}.
  \end{align*}
\end{definition}

\Cref{fig:mserpt} illustrates an example of
the relationship between the SERPT and \mserpt{} rank functions.
Under our assumptions on the job size distribution,
each of Gittins, SERPT, and \mserpt{}
has a continuous, piecewise monotonic rank function
\citep{mlps_gittins_aalto}.

\begin{figure}
  \centering
  \input{fig_mserpt}
  \vspace{-0.5\captionsqueeze}
  \caption{Example of SERPT and \mserpt{} Rank Functions}
  \label{fig:mserpt}
\end{figure}
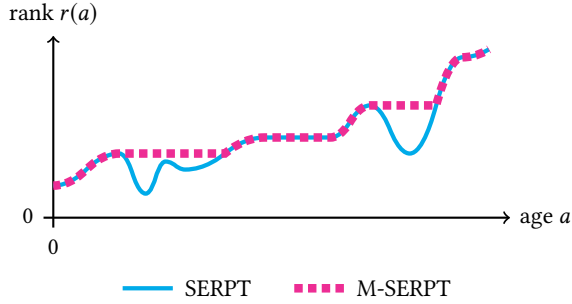

\section{Key Ideas}
\label{sec:outline}

We now give a high-level overview of how we prove our main result,
namely an upper bound on \mserpt{}'s approximation ratio.
The purpose of this section is to communicate, with minimal notation,
(1) the main ideas of our proof and
(2) the novelty of our approach.
As such, we discuss simplified versions of our key definitions and lemmas,
deferring the full versions to later in the paper.
For example, our main result in \cref{thm:approximation_ratio}
bounds \mserpt{}'s approximation ratio as a function of load,
but here we focus on a simpler corollary:
\begin{align}
  \label{eq:approximation_ratio_simple}
  \E{\response{}_\mserpt} \leq 5\E{\response{}_\gittins}.
\end{align}

\subsection{Waiting Time and Residence Time}
\label{sub:waiting_residence}

To prove \cref{eq:approximation_ratio_simple},
we first split response time into two pieces:
\begin{itemize}
\item
  \emph{residence time}~$\residence{}$,
  which is the response time
  of jobs that arrive to an empty system; and
\item
  \emph{waiting time}~$\waiting{}$,
  which is the extra delay due to the fact that the system is not always empty.
\end{itemize}
For SOAP policies,
response time is equal in distribution to
the independent sum of the waiting and residence times \citep{soap_scully}:
\begin{align*}
  \response{} = \waiting{} + \residence{}.
\end{align*}

The bound in \cref{eq:approximation_ratio_simple}
follows from two main lemmas,
one bounding each of \mserpt{}'s mean waiting and residence times.
Specifically, \cref{lem:waiting} implies
\begin{align}
  \label{eq:waiting_simple}
  \E{\waiting{}_\mserpt} \leq 2\E{\waiting{}_\gittins},
\end{align}
and \cref{lem:residence} implies
\begin{align}
  \label{eq:residence_simple}
  \E{\residence{}_\mserpt}
  \leq \E{\waiting{}_\mserpt} + \E{\response{}_\gittins}.
\end{align}
The proofs of \cref{eq:waiting_simple, eq:residence_simple}
constitute the main technical contribution of our work,
as their combination immediately yields \cref{eq:approximation_ratio_simple}:
\begin{align*}
  \E{\response{}_\mserpt}
  &= \E{\waiting{}_\mserpt} + \E{\residence{}_\mserpt} \\
  &\leq 2\E{\waiting{}_\mserpt} + \E{\response{}_\gittins} \\
  &\leq 4\E{\waiting{}_\gittins} + \E{\response{}_\gittins} \\
  &\leq 5\E{\response{}_\gittins}.
\end{align*}

\subsection{Why SOAP Is Not Enough}
\label{sub:why_not_soap}

How might we prove \cref{eq:waiting_simple, eq:residence_simple}?
One might think of using
the SOAP response time analysis of \citet{soap_scully}.
Their main result \citep[Theorem~5.5]{soap_scully}
takes a rank function $\rank{}$
and yields closed-form expressions
for $\E{\waiting{}}$ and $\E{\residence{}}$.
By ``closed-form'' expressions,
we mean functions of
the job size distribution's tail function~$\tail{}$ and the load~$\rho$
that can be written with just arithmetic and integrals.
However, the dependence on~$\rank{}$ is much more complicated.
This is a major obstacle for \mserpt{} and Gittins
because their rank functions depend on the job size distribution.
This makes it intractable to directly apply the SOAP analysis
to comparing \mserpt{} with Gittins over all job size distributions.

Much of the complexity of the SOAP analysis of \citet{soap_scully}
comes from being general enough to handle multiclass systems,
namely those in which different jobs follow different rank functions.
We only consider single-class systems in this paper.
Our approach is therefore
to simplify the SOAP analysis to our single-class setting
(\cref{sub:hills_valleys_outline}).
This results in much simpler expressions
for $\E{\waiting{}}$ and $\E{\residence{}}$,
partly because we are willing to settle for bounds.
The resulting simple expressions make it possible
to compare \mserpt{} to Gittins over all job size distributions
(\cref{sub:waiting_outline, sub:residence_outline}).

\subsection{Hills and Valleys}
\label{sub:hills_valleys_outline}

\begin{figure}
  \centering
  \input{fig_hills_valleys}
  \vspace{-\captionsqueeze}
  \caption{Hills and Valleys}
  \label{fig:hills_valleys}
\end{figure}
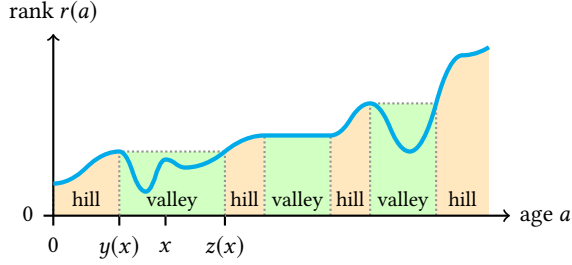

Suppose we are using a SOAP policy with rank function~$\rank{}$.
We use~$\rankup{}$, the increasing envelope of~$\rank{}$
(\cref{def:increasing_envelope}),
to classify ages into two types:
\begin{itemize}
\item
  \emph{hill ages},
  those at which $\rankup{}$ is strictly increasing; and
\item
  \emph{valley ages},
  those at which $\rankup{}$ is constant.
\end{itemize}
We call an interval of hill ages or valley ages
a \emph{hill} or \emph{valley}, respectively.
\Cref{fig:hills_valleys}, which shows an example of hills and valleys,
clarifies two points:
\begin{itemize}
\item
  Hill ages are those at which~$\rankup{}$,
  not just~$\rank{}$, is strictly increasing.
  For $a$ to be a hill age,
  not only must $\rank{}$ be increasing at age~$a$,
  but $\rank{}$ must not attain a greater rank at any earlier age.
\item
  Valley ages are those at which~$\rankup{}$,
  not~$\rank{}$, is constant.
  In general,
  $\rank{}$ might increase, decrease, or be constant at valley ages.
\end{itemize}
Given a size~$x$, we define two ages:
\begin{itemize}
\item
  the \emph{previous hill age}~$\y(x)$ is the greatest hill age $\leq x$, and
\item
  the \emph{next hill age}~$\z(x)$ is the least hill age $\geq x$.
\end{itemize}
If $x$ is a hill age, then $\y(x) = x = \z(x)$,
and if $x$ is a valley age, then $\y(x) < x < \z(x)$,
as illustrated in \cref{fig:hills_valleys}.\footnote{%
  We address some corner cases in the definitions
  of hills, valleys, $\y$, and~$\z$ in \cref{sub:hills_valleys_defs}.}

For any SOAP policy,
we can bound $\E{\waiting{}}$ and $\E{\residence{}}$
in terms of $\y$ and~$\z$.
\Cref{prop:waiting_hills_valleys} implies
\begin{align}
  \label{eq:waiting_hills_valleys}
  \E{\waiting{}}
  \geq \int_0^\infty
      \frac{\excess{\z(x)}}{\coload{\y(x)} \cdot \coload{\z(x)}}
    \density{x} \d{x},
\end{align}
and \cref{prop:residence_hills_valleys} implies
\begin{align}
  \label{eq:residence_hills_valleys}
  \E{\residence{}}
  \leq \int_0^\infty \frac{x}{\coload{\y(x)}} \density{x} \d{x},
\end{align}
with both bounds becoming equalities for \mserpt{}.
Here $\coload{}$ and~$\excess{}$ (\cref{def:load, def:excess})
are functions that do not depend on the scheduling policy.

Hills and valleys are important for two reasons.
First, the expressions in
\cref{eq:waiting_hills_valleys, eq:residence_hills_valleys}
depend on the scheduling policy only via $\y(x)$ and~$\z(x)$,
the previous and next hill ages of each size~$x$.
This means relating the mean response times of \mserpt{} and Gittins
partly reduces to relating the hills and valleys of \mserpt{} and Gittins.
Second, as we will soon see,
hills and valleys turn out to be important tools
for organizing the computations in the proofs of our two main bounds,
\cref{eq:waiting_simple, eq:residence_simple}.

\subsection{Outline of Waiting Time Bound}
\label{sub:waiting_outline}

We now outline the proof of \cref{eq:waiting_simple}, namely
$\E{\waiting[\mserpt]{}} \leq 2\E{\waiting[\gittins]{}}$.
By~\cref{eq:waiting_hills_valleys},
\begin{align}
  \label{eq:waiting_integral}
  \frac{\E{\waiting[\mserpt]{}}}{\E{\waiting[\gittins]{}}}
  \leq \cfrac{
      \displaystyle
      \int_0^\infty \frac{
          \excess{\z_\mserpt(x)}
        }{
          \coload{\y_\mserpt(x)} \cdot \coload{\z_\mserpt(x)}
        }
      \density{x} \d{x}
    }{
      \displaystyle
      \int_0^\infty \frac{
          \excess{\z_\gittins(x)}
        }{
          \coload{\y_\gittins(x)} \cdot \coload{\z_\gittins(x)}
        }
      \density{x} \d{x}
    }.
\end{align}
Our strategy for proving \cref{eq:waiting_simple}
is to split the integration regions in \cref{eq:waiting_integral} into chunks
and prove the bound for each chunk $[u, v]$:
\begin{align}
  \label{eq:waiting_chunk}
  \cfrac{
    \displaystyle
    \int_u^v \frac{
        \excess{\z_\mserpt(x)}
      }{
        \coload{\y_\mserpt(x)} \cdot \coload{\z_\mserpt(x)}
      }
    \density{x} \d{x}
  }{
    \displaystyle
    \int_u^v \frac{
        \excess{\z_\gittins(x)}
      }{
        \coload{\y_\gittins(x)} \cdot \coload{\z_\gittins(x)}
      }
    \density{x} \d{x}
  }
  \leq 2.
\end{align}
The key to this approach is to choose the right chunks.
It turns out that a good choice is for each Gittins hill and valley
to be a chunk.

As mentioned at the end of \cref{sub:hills_valleys_outline},
a key to comparing \mserpt{} to Gittins is comparing their hills and valleys.
We show in \cref{lem:hill_subset}
that every Gittins hill age is also an \mserpt{} hill age,
but not necessarily vice versa.
This implies that for any size~$x$,
\begin{align}
  \label{eq:age_ordering_simple}
  \y_\gittins(x) \leq \y_\mserpt(x)
  \leq x
  \leq \z_\mserpt(x) \leq \z_\gittins(x).
\end{align}

Proving \cref{eq:waiting_chunk} when chunk $[u, v]$ is a Gittins hill case
is simple.
When $x$ is a Gittins hill age,
\cref{eq:age_ordering_simple} collapses to an equality,
so the left-hand side of \cref{eq:waiting_chunk} is~$1$.

Proving \cref{eq:waiting_chunk} when chunk $[u, v]$ is a Gittins valley
is much more complicated.
As illustrated in \cref{fig:hills_valleys},
for all $x \in [u, v]$, we have $\y_\gittins(x) = u$ and $\z_\gittins(x) = v$,
simplifying the denominator in \cref{eq:waiting_chunk}.
Since $\excess{}$ is increasing (\cref{tab:monotonicity}),
it suffices to show $q(u, v) \leq 2$, where
\begin{align}
  \label{eq:waiting_ratio_chunk}
  q(a, b)
  = \int_a^b
      \frac{
        \coload{u} \cdot \coload{v}
      }{
        \coload{\y_\mserpt(x)} \cdot \coload{\z_\mserpt(x)}
      }
      \cdot \frac{\density{x}}{\tail{u} - \tail{v}}
    \d{x}.
\end{align}
We bound $q(u, v)$ by splitting it into
$q(u, v) = q(u, x_*) + q(x_*, v)$ for some~$x_*$.
Because the $\coload{}$~ratio in the integrand
is increasing in~$x$ (\cref{tab:monotonicity}),
the idea is to carefully choose $x_*$ such that
\begin{itemize}
\item
  $q(u, x_*) \leq 1$
  because the $\coload{}$~ratio is not too large for $x \in [u, x_*]$, and
\item
  $q(x_*, v) \leq 1$ because, roughly speaking,
  $f(x)$ is not too large for $x \in [x_*, v]$.
\end{itemize}
It turns out there is a natural choice for~$x_*$,
and the above strategy works when $x_*$ is an \mserpt{} hill age.
When $x_*$ is an \mserpt{} valley age,
we have to split $q(u, v)$ into three pieces instead,
with the third piece handling the valley containing~$x_*$,
but the upper bounds on the three pieces still add up to at most~$2$.

The proof of \cref{lem:waiting} in \cref{pf:waiting}
closely follows the strategy outlined in this section.
The main difference between \cref{eq:waiting_simple, lem:waiting}
is that the latter's bound is smaller at lower load.

\subsection{Outline of Residence Time Bound}
\label{sub:residence_outline}

We now outline the proof of \cref{eq:residence_simple}, namely
$\E{\residence[\mserpt]{}}
\leq \E{\waiting[\mserpt]{}} + \E{\response[\gittins]{}}$.
The first and more important step is \cref{lem:residence}, which says
\begin{align}
  \label{eq:residence_simple-ish}
  \E{\residence[\mserpt]{}}
  \leq \E{\waiting[\mserpt]{}}
    + \gp[\bigg]{\frac{1}{\rho}\log\frac{1}{1 - \rho}}\E{X}.
\end{align}
The second step uses a result of \citet[Theorem~5.8]{smart_insensitive_wierman}
to upper bound the last term in \cref{eq:residence_simple-ish}
by $\E{\response[\gittins]{}}$,
which yields \cref{eq:residence_simple}.

Our strategy for proving \cref{eq:residence_simple-ish} is, roughly speaking,
to integrate \cref{eq:waiting_hills_valleys, eq:residence_hills_valleys}
by parts:
\begin{align*}
  \E{\waiting[\mserpt]{}}
  &= \int_0^\infty
      \tail{x}
      \cdot\!
      \mathquote{
        \dd{x}
        \frac{
          \excess{\z_\mserpt(x)}
        }{
          \coload{\y_\mserpt(x)} \cdot \coload{\z_\mserpt(x)}
        }
      }\!
    \d{x} \\
  \E{\residence[\mserpt]{}}
  &= \int_0^\infty
      \tail{x}
      \cdot\!
      \mathquote{\dd{x} \frac{x}{\coload{\y_\mserpt(x)}}}
    \d{x}.
\end{align*}
These integral expressions are not rigorous
and are presented for intuition only.
Specifically, $\y_\mserpt$ and $\z_\mserpt$ have discontinuities,
so the derivatives are not well defined everywhere,
thus the quotation marks.
Again we split the integrals into chunks,
this time based on \mserpt{} hills and valleys,
and prove the bound for each chunk.

Proving the bound for \mserpt{} hills is simple
because $\y_\mserpt(x) = x = \z_\mserpt(x)$ when $x$ is an \mserpt{} hill age.
This means the derivatives are well defined,
and they even have a term in common, making them easy to compare.
In fact, we do not need any special properties of \mserpt{}
for this part of the argument.

Proving the bound for \mserpt{} valleys is more complicated.
Discontinuities of $\y_\mserpt$ and $\z_\mserpt$
occur at the boundaries of valleys.
Handling this requires some care,
but we nevertheless obtain simple expressions
for the waiting time and residence time chunks.
The main difficulty is that the expressions are difficult to compare.
It is this comparison that requires special properties of \mserpt{}.

The proof of \cref{lem:residence} in \cref{pf:residence}
closely follows the strategy outlined in this section.
However, when we put everything together to prove our main result,
\cref{thm:approximation_ratio},
it turns out jumping
from \cref{eq:residence_simple-ish} to \cref{eq:residence_simple}
is only a good idea at very high loads,
whereas using \cref{eq:residence_simple-ish} directly
yields a better bound at most loads.

\section{Hills and Valleys}
\label{sec:hills_valleys}

Hills and valleys are new concepts that play several important roles
in our bound of \mserpt{}'s approximation ratio
(\cref{sub:hills_valleys_outline, sub:waiting_outline, sub:residence_outline}).
The purpose of this section is to formally state
definitions and results relating to hills and valleys.
Throughout this section we work with
a generic SOAP policy with rank function~$r$.

\subsection{Defining Hills and Valleys}
\label{sub:hills_valleys_defs}

\begin{definition}
  \label{def:hill_age_valley_age}
  \leavevmode
  \begin{itemize}[topsep=0pt]
  \item
    A \emph{valley age} is an age $a > 0$ at which
    the increasing envelope~$\rankup{}$ of the rank function~$\rank{}$
    is locally constant,
    meaning there exists some $\epsilon > 0$ such that
    $\rankup{b} = \rankup{a}$ for all $b \in (a - \epsilon, a + \epsilon)$.
  \item
    A \emph{hill age} is an age that is not a valley age.
  \end{itemize}
\end{definition}

\begin{definition}
  \label{def:pred_succ}
  \leavevmode
  \begin{itemize}[topsep=0pt]
  \item
    The \emph{previous hill age} of size~$x$ is
    the latest hill age before~$x$:
    \begin{align*}
      \y(x) = \sup\{a < x \mid a \text{ is a hill age}\}.
    \end{align*}
  \item
    The \emph{next hill age} of size~$x$ is
    the earliest hill age after~$x$:\footnote{%
      There is a corner case for $x = 0$:
      we define $\y(0) = 0$ and $\z(0) = \z(0+)$,
      where postfix $+$ denotes a right limit.}
    \begin{align*}
      \z(x) = \inf\{a \geq x \mid a \text{ is a hill age}\}.
    \end{align*}
  \end{itemize}
\end{definition}

The difference in inequality strictness between $\y$ and $\z$
comes from how $\y$ and $\z$ are used to bound
mean waiting and residence times (\cref{app:soap_hills_valleys}).
For the most part,
$\y(x) = x = \z(x)$ for any hill age~$x$,
but there is an exception when $x$ is
preceded by an interval $(x - \epsilon, x)$ of valley ages.
This distinction is occasionally important,
so we extend our terminology to capture it.

\begin{definition}
  \label{def:hill_age_valley_age}
  \leavevmode
  \begin{itemize}[topsep=0pt]
  \item
    A \emph{hill size} is a size~$x$ such that $\y(x) = x = \z(x)$.
  \item
    A \emph{valley size} is a size that is not a hill size.
  \end{itemize}
\end{definition}

\begin{definition}
  \label{def:hills_valleys}
  \leavevmode
  \begin{itemize}[topsep=0pt]
  \item
    A \emph{hill} is an interval of hill sizes.
  \item
    A \emph{valley} is an interval $(u, v]$ of valley sizes
    where $u$ and~$v$ are hill ages.
  \end{itemize}
\end{definition}

\Cref{def:pred_succ, def:hills_valleys}
are illustrated in \cref{fig:hills_valleys}.
The distinction between hill ages and hill sizes
is important only for the upper boundaries of valleys,
which are hill ages but not hill sizes.\footnote{%
  There is another corner case for~$0$:
  it is always a hill age, but it is not a hill size if $\z(0+) > 0$.}

\subsection{Response Time Bounds}
\label{sub:hills_valleys_response_time}

We now use hills and valleys to write down simple bounds
on $\E{\waiting{x}}$ and $\E{\residence{x}}$,
the expected waiting and residence times (\cref{sub:waiting_residence}),
respectively,
of a job of size~$x$.

\begin{definition}
  \label{def:load}
  The \emph{$a$-truncated load complement} is
  one minus what the load of the system would be
  if every job's size were truncated at age~$a$:
  \begin{align*}
    \coload{a}
    = 1 - \lambda\E{\min\{X, a\}}
    = 1 - \int_0^a \lambda \tail{t} \d{t}.
  \end{align*}
\end{definition}

\begin{definition}
  \label{def:excess}
  The \emph{$a$-truncated second moment factor} is
  \begin{align*}
    \excess{a}
    = \frac{\lambda}{2}\E{(\min\{X, a\})^2}
    = \int_0^x \lambda t \tail{t} \d{t}.
  \end{align*}
\end{definition}

\begin{proposition}
  \label{prop:waiting_hills_valleys}
  Under any SOAP policy, the expected waiting time of a job of size~$x$
  is bounded by
  \begin{align*}
    \E{\waiting{x}}
    \geq \frac{\excess{\z(x)}}{\coload{\y(x)} \cdot \coload{\z(x)}},
  \end{align*}
  with equality if the policy has a monotonic rank function.
\end{proposition}

\begin{proposition}
  \label{prop:residence_hills_valleys}
  Under any SOAP policy, the expected residence time of a job of size~$x$
  is bounded by
  \begin{align*}
    \E{\residence{x}}
    \leq \frac{x}{\coload{\y(x)}},
  \end{align*}
  with equality if the policy has a monotonic rank function.
\end{proposition}

\proofin[prop:waiting_hills_valleys,
  prop:residence_hills_valleys]{app:soap_hills_valleys}

\mserpt{} has a monotonic rank function,
so both \cref{prop:waiting_hills_valleys, prop:residence_hills_valleys}
yield useful equalities for \mserpt{}.
However,
to prove an upper bound on \mserpt{}'s approximation ratio,
we want lower bounds for Gittins,
for which only \cref{prop:waiting_hills_valleys} is useful.
Instead of using \cref{prop:residence_hills_valleys} for Gittins,
we use the following lower bounds.

\begin{proposition}
  \label{prop:residence_size}
  Under any SOAP policy, the mean residence time is bounded by
  $\E{\residence{}} \geq \E{X}$.
\end{proposition}

\begin{proof}
  A job's residence time is, by definition (\cref{sub:waiting_residence}),
  at least its size.
\end{proof}

\begin{proposition}
  \label{prop:response_srpt}
  Under any scheduling policy, the mean response time is bounded by
  \begin{align*}
    \E{\response{}} \geq \gp[\bigg]{\frac{1}{\rho}\log\frac{1}{1 - \rho}}\E{X}.
  \end{align*}
\end{proposition}

\begin{proof}
  \Citet[Theorem~5.8]{smart_insensitive_wierman} show that
  the desired lower bound holds for SRPT,
  which has lower mean response time than any other policy
  \citep{srpt_optimal_schrage}.
\end{proof}

\section{Upper Bound on \mserpt's Approximation Ratio}
\label{sec:approximation_ratio}

\begin{figure}
  \centering
  \input{fig_outline}
  \vspace{-\captionsqueeze}
  \caption{Bounding Mean Response Time of \mserpt{}}
  \label{fig:outline}
\end{figure}
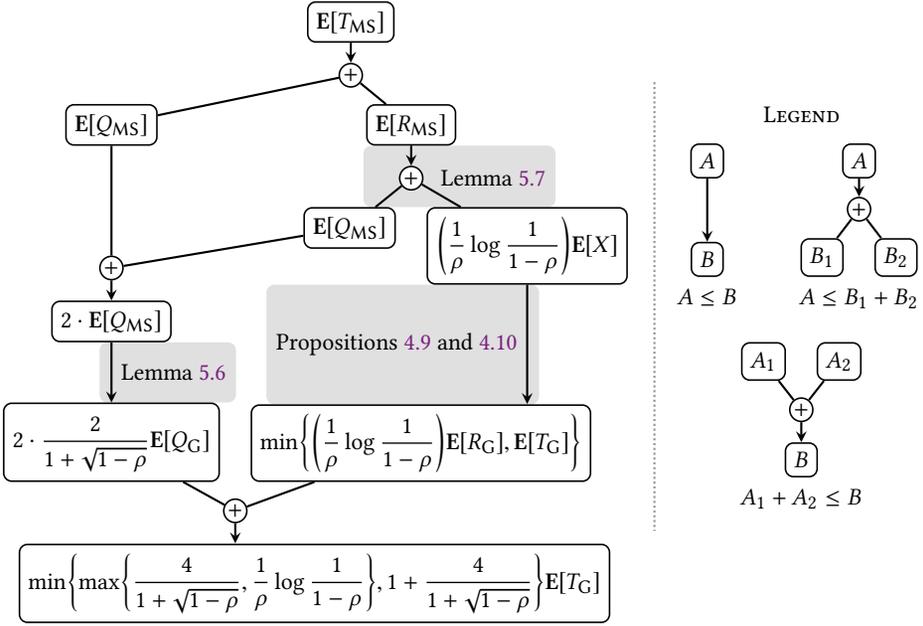

In this section we prove our main result,
which is an upper bound on
the mean response time ratio between \mserpt{} and Gittins.

\begin{restatable}{theorem}{oldthmapproximationratio}
  \label{thm:approximation_ratio}
  The mean response time ratio between \mserpt{} and Gittins
  is bounded by\kern 0.083333333em\footnote{%
    The numbers $0.9587$ and $0.9898$ are approximations
    accurate to $4$ decimal places.}
  \begin{align*}
    \frac{\E{\response[\mserpt]{}}}{\E{\response[\gittins]{}}} \leq
    \begin{dcases}
      \frac{4}{1 + \sqrt{1 - \rho}} & 0 \leq \rho < 0.9587 \\
      \frac{1}{\rho}\log\frac{1}{1 - \rho} & 0.9587 \leq \rho < 0.9898 \\
      1 + \frac{4}{1 + \sqrt{1 - \rho}} & 0.9898 \leq \rho < 1.
    \end{dcases}
  \end{align*}
\end{restatable}
\newcommand{\thmapproximationratio}{{%
  \renewcommand{\footnote}[1]{}\oldthmapproximationratio*}}

\proofin{fig:outline, pf:approximation_ratio}

As illustrated in \cref{fig:outline},
the main steps in the proof of \cref{thm:approximation_ratio}
are \cref{lem:waiting, lem:residence} (\cref{sub:waiting, sub:residence}).
\Cref{fig:ratio_bound} plots the resulting bound as a function of load~$\rho$.
The following corollary gives intuition for this function
in terms of concrete values.

\begin{corollary}
  \label{cor:approximation_ratio}
  For the problem of preemptive scheduling to minimize mean response time
  in an M/G/1 queue with unknown job sizes,
  the approximation ratio of \mserpt{} is at most
  \begin{itemize}
  \item
    $2.5$ for load $\rho \leq 0.64$,
  \item
    $3$ for load $\rho \leq 8/9 \approx 0.89$,
  \item
    $3.3$ for load $\rho \leq 0.95$,
  \item
    $4$ for load $\rho \leq 0.98$, and
  \item
    $5$ for all loads.
\end{itemize}
\end{corollary}

\begin{figure}
  \centering
  \input{fig_ratio_bound}
  \vspace{-\captionsqueeze}
  \caption{Bound on Mean Response Time Ratio}
  \label{fig:ratio_bound}
\end{figure}
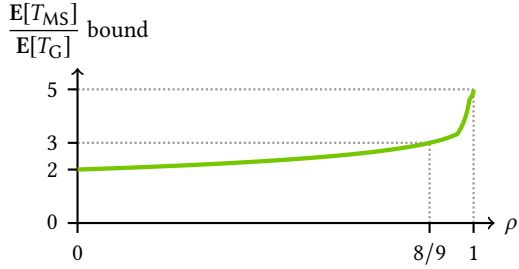

\subsection{Properties of \mserpt{} Hill Ages}

In this section we prove some properties of \mserpt{} hills and valleys,
and in particular \mserpt{} hill ages.
We begin by relating the hills and valleys of \mserpt{} and Gittins.
The following lemma builds on ideas introduced by \citet{mlps_gittins_aalto},
but it is a novel result.\footnote{%
  In particular, \cref{lem:hill_subset} is not equivalent to
  Proposition~7 of \citet{mlps_gittins_aalto}
  because hills are not simply the ages
  at which the rank function is increasing
  (\cref{def:hills_valleys}).}

\begin{restatable}{lemma}{lemhillsubset}
  \label{lem:hill_subset}
  Every Gittins hill age is also an \mserpt{} hill age,
  and similarly for hill sizes.
\end{restatable}

\proofin{pf:hill_subset}

We now show a key property of \mserpt{} hill ages
that lets us to bound $\coload{}$~ratios,
such as those in~\cref{eq:waiting_ratio_chunk},
in terms of $\tail{}$~ratios.

\begin{lemma}
  \label{lem:coload}
  For any \mserpt{} hill age~$b$ and any $a \leq b$,
  \begin{align*}
    \frac{\coload{a}}{\coload{b}}
    \leq \frac{1}{1 - \rho + \rho \frac{\tail{b}}{\tail{a}}}.
  \end{align*}
\end{lemma}

\begin{proof}
  Recall from \cref{def:mserpt} that
  $\rank[\mserpt]{}$ is the increasing envelope of $\rank[\serpt]{}$.
  By \cref{def:hill_age_valley_age},
  this means \mserpt{} has the same hill and valley ages as SERPT.
  We therefore have
  \begin{itemize}
  \item
    $\rank[\serpt]{a} \leq \rank[\mserpt]{a}$
    by \cref{def:mserpt},
  \item
    $\rank[\mserpt]{a} \leq \rank[\mserpt]{b}$
    because $\rank[\mserpt]{}$ is increasing, and
  \item
    $\rank[\serpt]{b} = \rank[\mserpt]{b}$
    because $b$ is a SERPT hill age.
  \end{itemize}
  Putting these together gives us $\rank[\serpt]{a} \leq \rank[\serpt]{b}$,
  which by \cref{def:serpt} is the same as
  \begin{align*}
    \frac{\int_a^\infty \tail{t} \d{t}}{\tail{a}}
    \leq \frac{\int_b^\infty \tail{t} \d{t}}{\tail{b}}.
  \end{align*}
  Multiplying both sides by $\lambda$ and applying \cref{def:load} yields
  \begin{align*}
    \frac{\coload{a} - (1 - \rho)}{\tail{a}}
    \leq \frac{\coload{b} - (1 - \rho)}{\tail{b}}.
  \end{align*}
  Letting $\zeta = \tail{b}/\tail{a}$, this rearranges to
  \begin{align*}
    \frac{\coload{b}}{\coload{a}}
    \geq \zeta + (1 - \zeta)\frac{1 - \rho}{\coload{a}}.
  \end{align*}
  Because $\coload{a} \leq 1$,
  the right-hand side is at least $1 - \rho + \zeta\rho$,
  which implies the desired inequality.
\end{proof}

The bound in \cref{lem:coload} is increasing in~$\rho$,
implying the following simpler bound.

\begin{corollary}
  \label{cor:coload}
  For any \mserpt{} hill age~$b$ and any $a \leq b$,
  \begin{align*}
    \frac{\coload{a}}{\coload{b}} \leq \frac{\tail{a}}{\tail{b}}.
  \end{align*}
\end{corollary}

\subsection{Waiting Time Bound}
\label{sub:waiting}

The proofs in the remainder of this section frequently use
the monotonicity facts listed in \cref{tab:monotonicity}.
As a reminder, all monotonicities are meant in the weak sense
unless otherwise specified.
For example, ``decreasing'' means nonincreasing.
So as not to disrupt the flow of the proofs,
we use facts from \cref{tab:monotonicity} with only a reference to the table.

\begin{table}
  \centering
  \caption{Monotonicity Facts}
  \label{tab:monotonicity}
  \vspace{-0.5\captionsqueeze}
  \begin{tabular}{@{}lll@{}}
    \toprule
    \textsc{Function} & \textsc{Monotonicity} & \textsc{Defined in} \\
    \midrule
    $\tail{}$ & decreasing & \cref{sec:model} \\
    $\coload{}$ & decreasing & \cref{def:load} \\
    $\excess{}$ & increasing & \cref{def:excess} \\
    $\y, \z$ & increasing & \cref{def:pred_succ} \\
    \bottomrule
  \end{tabular}
\end{table}

\begin{lemma}
  \label{lem:waiting}
  The mean waiting time of \mserpt{} is bounded by
  \begin{align*}
    \frac{\E{\waiting[\mserpt]{}}}{\E{\waiting[\gittins]{}}}
    \leq \frac{2}{1 + \sqrt{1 - \rho}}.
  \end{align*}
\end{lemma}

\begin{proof}
  \label{pf:waiting}
  By \cref{lem:hill_subset},
  because $\y_\gittins(x) = x = \z_\gittins(x)$
  for all Gittins hill sizes~$x$,
  we have
  \begin{align*}
    \frac{
      \E{\waiting[\mserpt]{X} \given
        X \text{ is a Gittins hill size}}
    }{
      \E{\waiting[\gittins]{X} \given
        X \text{ is a Gittins hill size}}
    }
    \leq 1.
  \end{align*}
  Therefore, it suffices to show that
  for any Gittins valley $(u, v]$,
  \begin{align*}
    \frac{
      \E{\waiting[\mserpt]{X} \given X \in (u, v]}
    }{
      \E{\waiting[\gittins]{X} \given X \in (u, v]}
    }
    \leq \frac{2}{1 + \sqrt{1 - \rho}}.
  \end{align*}

  For any $x \in (u, v]$,
  \cref{lem:hill_subset} implies the following key fact:
  \begin{align}
    \label{eq:age_ordering}
    u = \y_\gittins(x) \leq \y_\mserpt(x)
    \leq x
    \leq \z_\mserpt(x) \leq \z_\gittins(x) = v.
  \end{align}
  Applying \cref{prop:waiting_hills_valleys, tab:monotonicity},
  we obtain
  \begin{align*}
    \iftwocol{\MoveEqLeft}{}
    \frac{
      \E{\waiting[\mserpt]{X} \given X \in (u, v]}
    }{
      \E{\waiting[\gittins]{X} \given X \in (u, v]}
    }
    \iftwocol{\\}{}
    &= \cfrac{
        \displaystyle
        \int_u^v
          \frac{
            \excess{\z_\mserpt(x)}
          }{
            \coload{\y_\mserpt(x)} \cdot \coload{\z_\mserpt(x)}
          }
          \cdot \frac{\density{x}}{\tail{u} - \tail{v}}
        \d{x}
      }{
        \displaystyle
        \frac{\excess{v}}{\coload{u} \cdot \coload{v}}
      } \\[0.15em]
    &\leq \int_u^v
      \frac{
        \coload{u} \cdot \coload{v}
      }{
        \coload{\y_\mserpt(x)} \cdot \coload{\z_\mserpt(x)}}
        \cdot \frac{\density{x}}{\tail{u} - \tail{v}
      }
      \d{x}.
  \end{align*}
  Let
  \begin{align*}
    q(a, b)
    = \int_a^b
        \frac{
          \coload{u} \cdot \coload{v}
        }{
          \coload{\y_\mserpt(x)} \cdot \coload{\z_\mserpt(x)}
        }
        \cdot \frac{\density{x}}{\tail{u} - \tail{v}}
      \d{x}.
  \end{align*}
  It suffices to bound $q(u, v)$.
  To do so, we split the integration region
  into three pieces at carefully chosen ages $\y_*$ and~$\z_*$,
  then we bound each of
  $q(u, \y_*)$, $q(\y_*, \z_*)$, and $q(\z_*, v)$.

  Before specifying $\y_*$ and~$\z_*$, we need two other definitions.
  First, for all $x \in (u, v]$, let
  \begin{align*}
    \gail{x} = 1 - \rho + \rho \frac{\tail{x}}{\tail{u}}.
  \end{align*}
  With this notation, \cref{lem:coload} says that
  if $x$ is an \mserpt{} hill age, then\footnote{%
    Even though $\gail{u} = 1$,
    we find that explicitly writing $\gail{u}$
    in ratios with other uses of $\gail{}$
    makes the proof easier to follow.}
  \begin{align}
    \label{eq:coload}
    \frac{\coload{u}}{\coload{x}}
    \leq \frac{\gail{u}}{\gail{x}}.
  \end{align}
  Second, let $x_* \in (u, v]$ be the age such that
  \begin{align}
    \label{eq:xstar}
    \frac{\gail{u}}{\gail{x_*}}
    = \frac{\coload{u}}{\coload{v}}.
  \end{align}
  Such an age must exist by continuity of~$\gail{}$
  because by \cref{tab:monotonicity} and~\cref{eq:coload},
  \begin{align*}
    \gail{u}
    \geq \frac{\coload{v}}{\coload{u}}\gail{u}
    \geq \gail{v}.
  \end{align*}

  We can now define
  \begin{align*}
    \y_* &= \y_\mserpt(x_*) \\
    \z_* &= \z_\mserpt(x_*).
  \end{align*}
  We bound each of
  $q(u, \y_*)$, $q(\y_*, \z_*)$, and $q(\z_*, v)$
  in \cref{step:prestar, step:star, step:poststar} below.
  The core of each step is bounding the ratios
  $\coload{u}/\coload{\y_\mserpt(x)}$
  and $\coload{v}/\coload{\z_\mserpt(x)}$.
  \begin{itemize}
  \item
    By \cref{tab:monotonicity, eq:age_ordering} we have
    \begin{align}
      \label{eq:coload_z_poststar}
      \frac{\coload{v}}{\coload{\z_\mserpt(x)}} \leq 1
    \end{align}
    and, using~\cref{eq:xstar},
    \begin{align}
      \label{eq:coload_y_poststar}
      \frac{\coload{u}}{\coload{\y_\mserpt(x)}}
      = \frac{\coload{v}}{\coload{\y_\mserpt(x)}}
        \cdot \frac{\gail{u}}{\gail{x_*}}
      \leq \frac{\gail{u}}{\gail{x_*}}.
    \end{align}
  \item
    Since $\y_\mserpt(x)$ and $\z_\mserpt(x)$ are \mserpt{} hill ages,
    by \cref{eq:coload} we have
    \begin{align}
      \label{eq:coload_y_prestar}
      \frac{\coload{u}}{\coload{\y_\mserpt(x)}}
      \leq \frac{\gail{u}}{\gail{\y_\mserpt(x)}}
    \end{align}
    and, using~\cref{eq:xstar},
    \begin{align}
      \label{eq:coload_z_prestar}
      \frac{\coload{v}}{\coload{\z_\mserpt(x)}}
      = \frac{\coload{u}}{\coload{\z_\mserpt(x)}}
        \cdot \frac{\gail{x_*}}{\gail{u}}
      \leq \frac{\gail{x_*}}{\gail{\z_\mserpt(x)}}.
    \end{align}
  \end{itemize}
  In each of \cref{step:prestar, step:star, step:poststar},
  we apply either \cref{eq:coload_z_poststar} or~\cref{eq:coload_z_prestar},
  whichever gives a tighter bound,
  and similarly for \cref{eq:coload_y_poststar, eq:coload_y_prestar}.

  We need one last definition
  before carrying out \cref{step:prestar, step:star, step:poststar}:
  to avoid mixing $\tail{}$ and~$\gail{}$, let
  \begin{align*}
    \gensity{x} = -\dd{x} \gail{x} = \rho\density{x},
  \end{align*}
  which allows us to write
  \begin{align*}
    \frac{\density{x}}{\tail{u} - \tail{v}}
    = \frac{\gensity{x}}{\gail{u} - \gail{v}}.
  \end{align*}

  \setcounter{step}{0}
  \begin{step}[{bounding $q(u, \y_*)$}]
    \label{step:prestar}
    Since both $u$ and $\y_*$ are \mserpt{} hill ages,
    we can partition $(u, \y_*]$ into \mserpt{} hills and valleys,\footnote{%
      The potential obstacle to partitioning is that $u$ or $\y_*$ might be
      in the interior of a valley (\cref{def:hills_valleys}),
      but $u$ and $\y_*$ being hill ages ensures this is not the case.}
    meaning there exist
    \begin{align*}
      u = \z_0 \leq \y_1 < \z_1 < \dots
      < \y_n < \z_n \leq \y_{n + 1} = \y_*
    \end{align*}
    such that
    \begin{itemize}
    \item
      $(\y_i, \z_i]$ is an \mserpt{} valley for all $i \in \{1, \dots, n\}$,
    \item
      $(\z_i, \y_{i + 1}]$ is an \mserpt{} hill
      for all $i \in \{1, \dots, n\}$, and
    \item
      either $\z_0 = \y_1$ or $(\z_0, \y_1]$ is an \mserpt{} hill.
    \end{itemize}
    For each \mserpt{} valley,
    we have $\y_\mserpt(x) = \y_i$ and $\z_\mserpt(x) = \z_i$
    for $x \in (\y_i, \z_i]$,
    so applying \cref{eq:coload_y_prestar, eq:coload_z_prestar}
    yields
    \begin{align}
      \label{eq:prestar_valley}
      q(\y_i, \z_i)
      &\leq \int_{\y_i}^{\z_i}
          \frac{\gail{u} \cdot \gail{x_*}}{\gail{\y_i} \cdot \gail{\z_i}}
          \cdot
          \frac{\gensity{x}}{\gail{u} - \gail{v}}
        \d{x} \\
      &= \frac{\gail{u} \cdot \gail{x_*}}{\gail{u} - \gail{v}}
        \gp[\bigg]{\frac{1}{\gail{\z_i}} - \frac{1}{\gail{\y_i}}}.
    \end{align}
    For each \mserpt{} hill,
    we have $\y_\mserpt(x) = x = \z_\mserpt(x)$ for $x \in (\z_i, \y_{i + 1}]$,
    so applying \cref{eq:coload_y_prestar, eq:coload_z_prestar} yields
    \begin{align}
      \label{eq:prestar_hill}
      q(\z_i, \y_{i + 1})
      &\leq \int_{\z_i}^{\y_{i + 1}}
          \frac{\gail{u} \cdot \gail{x_*}}{\gail{x}^2}
          \cdot
          \frac{\gensity{x}}{\gail{u} - \gail{v}}
        \d{x} \\
      &= \frac{\gail{u} \cdot \gail{x_*}}{\gail{u} - \gail{v}}
        \gp[\bigg]{\frac{1}{\gail{\y_{i + 1}}} - \frac{1}{\gail{\z_i}}}.
    \end{align}
    Combining \cref{eq:prestar_valley, eq:prestar_hill}
    for each \mserpt{} hill and valley implies
    \begin{align*}
      q(u, \y_*)
      &= \sum_{i = 1}^n q(\y_i, \z_i)
        + \sum_{i = 0}^n q(\z_i, \y_{i + 1}) \\
      &\leq \frac{\gail{u} \cdot \gail{x_*}}{\gail{u} - \gail{v}}
        \gp[\bigg]{\frac{1}{\gail{\y_{n + 1}}} - \frac{1}{\gail{\z_0}}} \\
      &= \frac{\gail{u}}{\gail{u} - \gail{v}}
        \gp[\bigg]{\frac{\gail{x_*}}{\gail{\y_*}}
          - \frac{\gail{x_*}}{\gail{u}}}.
    \end{align*}
  \end{step}

  \begin{step}[{bounding $q(\y_*, \z_*)$}]
    \label{step:star}
    If $x_*$ is an \mserpt{} hill size,
    then $q(\y_*, \z_*) = q(x_*, x_*) = 0$.
    Otherwise,
    since $\y_\mserpt(x) = \y_*$ for all $x \in (\y_*, \z_*]$,
    applying \cref{eq:coload_z_poststar, eq:coload_y_prestar} yields
    \begin{align*}
      q(\y_*, \z_*)
      &\leq \int_{\y_*}^{\z_*}
          \frac{\gail{u}}{\gail{\y_*}}
          \cdot
          \frac{\gensity{x}}{\gail{u} - \gail{v}}
        \d{x} \\
      &= \frac{\gail{u}}{\gail{u} - \gail{v}}
        \gp[\bigg]{1 - \frac{\gail{z_*}}{\gail{\y_*}}}.
    \end{align*}
  \end{step}

  \begin{step}[{bounding $q(\z_*, v)$}]
    \label{step:poststar}
    Applying \cref{eq:coload_z_poststar, eq:coload_y_poststar} yields
    \begin{align*}
      q(\z_*, v)
      &\leq \int_{\z_*}^v
          \frac{\gail{u}}{\gail{x_*}}
          \cdot
          \frac{\gensity{x}}{\gail{u} - \gail{v}}
        \d{x} \\
      &= \frac{\gail{u}}{\gail{u} - \gail{v}}
        \gp[\bigg]{\frac{\gail{\z_*}}{\gail{x_*}}
          - \frac{\gail{v}}{\gail{x_*}}}.
    \end{align*}
  \end{step}

  Combining the results of \cref{step:prestar, step:star, step:poststar}
  gives us
  \begin{align}
    \label{eq:q_bound}
    q(u, v)
    \iftwocol{&}{}
    \leq \frac{\gail{u}}{\gail{u} - \gail{v}}
      \biggl(
        \frac{\gail{x_*}}{\gail{\y_*}}
        - \frac{\gail{x_*}}{\gail{u}}
    \iftwocol{\\ &\qquad}{}
        + 1
        - \frac{\gail{z_*}}{\gail{\y_*}}
        + \frac{\gail{\z_*}}{\gail{x_*}}
        - \frac{\gail{v}}{\gail{x_*}}\biggr).
  \end{align}
  \Cref{tab:monotonicity, eq:age_ordering} imply
  \begin{align*}
    \frac{\gail{\z_*}}{\gail{x_*}} - \frac{\gail{z_*}}{\gail{\y_*}}
    &\leq 1 - \frac{\gail{x_*}}{\gail{\y_*}},
  \end{align*}
  and minimizing over possible values of $\gail{x_*}$ gives
  \begin{align*}
    \frac{\gail{x_*}}{\gail{u}} + \frac{\gail{v}}{\gail{x_*}}
    &\geq 2\sqrt{\frac{\gail{v}}{\gail{u}}}.
  \end{align*}
  Applying these to \cref{eq:q_bound}
  and using the fact that $\gail{v}/\gail{u} \geq 1 - \rho$
  yields
  \begin{align*}
    q(u, v)
    &\leq \frac{\gail{u}}{\gail{u} - \gail{v}}
      \gp[\Bigg]{2 - 2\sqrt{\frac{\gail{v}}{\gail{u}}}} \\
    &= \frac{2}{1 + \sqrt{\frac{\gail{v}}{\gail{u}}}} \\
    &\leq \frac{2}{1 + \sqrt{1 - \rho}}.
    \qedhere
  \end{align*}
\end{proof}

\subsection{Residence Time Bound}
\label{sub:residence}

\begin{lemma}
  \label{lem:residence}
  The mean residence time of \mserpt{} is bounded by
  \begin{align*}
    \E{\residence[\mserpt]{}}
    \leq \E{\waiting[\mserpt]{}}
      + \gp[\bigg]{\frac{1}{\rho}\log\frac{1}{1 - \rho}}\E{X}.
  \end{align*}
\end{lemma}

\begin{proof}
  \label{pf:residence}
  We can partition $\R_{\geq 0}$ into \mserpt{} hills and valleys,
  meaning there exist
  \begin{align*}
    0 = \z_0 < \y_1 < \z_1 < \dots
  \end{align*}
  such that
  \begin{itemize}
  \item
    $(\y_i, \z_i]$ is an \mserpt{} valley for all $i \geq 1$,
  \item
    $(\z_i, \y_{i + 1}]$ is an \mserpt{} hill for all $i \geq 1$, and
  \item
    either $\z_0 = \y_0$ or $(\z_0, \y_1]$ is an \mserpt{} hill.
  \end{itemize}

  Let
  \begin{align*}
    \Delta_Q(a, b)
    &= \E{\waiting[\mserpt]{\min\{X, b\}}}
      - \E{\waiting[\mserpt]{\min\{X, a\}}} \\
    \Delta_R(a, b)
    &= \E{\residence[\mserpt]{\min\{X, b\}}}
      - \E{\residence[\mserpt]{\min\{X, a\}}} \\
    \Delta_{\log}(a, b)
    &= \frac{1}{\lambda}\log\frac{1}{\coload{b}}
      - \frac{1}{\lambda}\log\frac{1}{\coload{a}}.
  \end{align*}
  We wish to show
  $\Delta_R(0, \infty) \leq \Delta_Q(0, \infty) + \Delta_{\log}(0, \infty)$.
  It suffices to show that
  for each \mserpt{} hill $(\z_i, \y_{i + 1}]$,\footnote{%
    We use postfix $-$ and~$+$ to denote left and right limits,
    respectively.
    They are not needed for $\Delta_{\log}$,
    which is continuous.}
  \begin{align}
    \label{eq:delta_hill}
    \Delta_R(\z_i+, \y_{i + 1}-)
    \leq \Delta_Q(\z_i+, \y_{i + 1}-) + \Delta_{\log}(\z_i, \y_{i + 1}),
  \end{align}
  and that for each \mserpt{} valley $(\y_i, \z_i]$,
  \begin{align}
    \label{eq:delta_valley}
    \Delta_R(\y_i-, \z_i+)
    \leq \Delta_Q(\y_i-, \z_i+) + \Delta_{\log}(\y_i, \z_i).
  \end{align}
  We prove these bounds
  in \cref{step:delta_hill, step:delta_valley} below, respectively.
  In both steps we use the fact that
  \begin{align*}
    \dd{x}\E{\waiting[\mserpt]{\min\{X, x\}}}
    &= \tail{x} \cdot \dd{x}\E{\waiting[\mserpt]{x}} \\
    \dd{x}\E{\residence[\mserpt]{\min\{X, x\}}}
    &= \tail{x} \cdot \dd{x}\E{\residence[\mserpt]{x}}.
  \end{align*}

  \setcounter{step}{0}
  \begin{step}[{bound for \mserpt{} hills}]
    \label{step:delta_hill}
    We have $\y_\mserpt(x) = x = \z_\mserpt(x)$
    for all $x \in (\z_i, \y_{i + 1})$.
    Recalling \cref{def:load, def:excess},
    by \cref{prop:waiting_hills_valleys},
    \begin{align*}
      \dd{x}\Delta_Q(\z_i+, x)
      &= \dd{x}\E{\waiting[\mserpt]{\min\{X, x\}}} \\
      &= \tail{x} \cdot \dd{x}\frac{\excess{x}}{\coload{x}^2} \\
      &= \frac{\lambda x \tail{x}^2}{\coload{x}^2}
        + \frac{2\lambda \tail{x}^2 \cdot \excess{x}}{\coload{x}^3}.
    \end{align*}
    Similarly, by \cref{prop:residence_hills_valleys},
    \begin{align*}
      \dd{x}\Delta_R(\z_i+, x)
      &= \dd{x}\E{\residence[\mserpt]{\min\{X, x\}}} \\
      &= \tail{x} \cdot \dd{x}\frac{x}{\coload{x}} \\
      &= \frac{\tail{x}}{\coload{x}}
        + \frac{\lambda x \tail{x}^2}{\coload{x}^2}.
    \end{align*}
    Finally, we have
    \begin{align*}
      \dd{x}\Delta_{\log}(\z_i, x) = \frac{\tail{x}}{\coload{x}}.
    \end{align*}
    Examining the three derivatives, we see
    \begin{align*}
      \dd{x}\Delta_R(\z_i+, x)
      \leq \dd{x}\Delta_Q(\z_i+, x) + \dd{x}\Delta_{\log}(\z_i, x),
    \end{align*}
    which implies~\cref{eq:delta_hill}, as desired.
  \end{step}

  \begin{step}[{bound for \mserpt{} valleys}]
    \label{step:delta_valley}
    We have $\y_\mserpt(x) = \y_i$ and $\z_\mserpt(x) = \z_i$
    for all $x \in (\y_i, \z_i)$,
    which means
    \begin{align*}
      \dd{x}\Delta_Q(\y_i-, x) &= 0 \\
      \dd{x}\Delta_R(\y_i-, x) &= \frac{\tail{x}}{\coload{\y_i}} \\
      \dd{x}\Delta_{\log}(\y_i, x) &= \frac{\tail{x}}{\coload{x}}.
    \end{align*}
    However, we must still account for discontinuities
    at $x = \y_i$ and $x = \z_i$.

    We first prove a lower bound on $\Delta_Q(\y_i-, \z_i+)$.
    We have
    \begin{align}
      \label{eq:delta_q_progress}
      \Delta_Q(\y_i-, \z_i+)
      &= \Delta_Q(\y_i-, \y_i+) + \Delta_Q(\z_i-, \z_i+) \\
      &= \tail{y_i} \gp[\bigg]{%
          \frac{\excess{\z_i}}{\coload{\y_i} \cdot \coload{\z_i}}
          - \frac{\excess{\y_i}}{\coload{\y_i}^2}}
      \iftwocol{\\ &\qquad}{}
        + \tail{\z_i} \gp[\bigg]{%
          \frac{\excess{\z_i}}{\coload{\z_i}^2}
          - \frac{\excess{\z_i}}{\coload{\y_i} \cdot \coload{\z_i}}}.
    \end{align}
    Both terms in \cref{eq:delta_q_progress} are nonnegative
    by \cref{tab:monotonicity}.
    Applying \cref{cor:coload} with $a = \y_i$ and $b = \z_i$
    to the first term and dropping the second term yields
    \begin{align}
      \label{eq:delta_q_bound}
      \Delta_Q(\y_i-, \z_i+)
      \geq \tail{\z_i} \gp[\bigg]{%
        \frac{\excess{\z_i}}{\coload{\z_i}^2}
        - \frac{\excess{\y_i}}{\coload{\y_i} \cdot \coload{\z_i}}}.
    \end{align}

    We now turn to $\Delta_R(\y_i-, \z_i+)$.
    We have
    \begin{align}
      \label{eq:delta_r_progress}
      \Delta_R(\y_i-, \z_i+)
      &= \Delta_R(\y_i-, \y_i+) + \Delta_R(\y_i+, \z_i-)
        + \Delta_R(\z_i-, \z_i+) \\
      &= 0 + \int_{\y_i}^{\z_i} \frac{\tail{x}}{\coload{\y_i}} \d{x}
        + \tail{\z_i} \gp[\bigg]{%
          \frac{\z_i}{\coload{\z_i}} - \frac{\z_i}{\coload{\y_i}}} \\
      &= \Delta_{\log}(\y_i, \z_i)
        + \tail{\z_i} \gp[\bigg]{%
          \frac{\z_i}{\coload{\z_i}} - \frac{\z_i}{\coload{\y_i}}}
      \iftwocol{\\ &\qquad}{}
        - \int_{\y_i}^{\z_i}
          \tail{x} \gp[\bigg]{%
            \frac{1}{\coload{x}} - \frac{1}{\coload{\y_i}}}
          \d{x}.
    \end{align}
    Applying \cref{cor:coload} with $a = x$ and $b = \z_i$
    to the last term of \cref{eq:delta_r_progress} yields
    \begin{align}
      \label{eq:delta_r_bound}
      \Delta_R(\y_i-, \z_i+)
      \iftwocol{&}{}
      \leq \Delta_{\log}(\y_i, \z_i)
        + \tail{\z_i} \gp[\bigg]{%
          \frac{\z_i}{\coload{\z_i}} - \frac{\z_i}{\coload{\y_i}}}
      \iftwocol{\\ &\qquad}{}
        - \int_{\y_i}^{\z_i}
          \tail{\z_i} \gp[\bigg]{%
            \frac{1}{\coload{\z_i}}
            - \frac{\coload{x}}{\coload{\y_i} \cdot \coload{\z_i}}}
          \d{x}.
    \end{align}
    Using integration by parts one can compute
    \begin{align*}
      \int_{\y_i}^{\z_i} \coload{x} \d{x}
      = \z_i\coload{\z_i} - \y_i\coload{\y_i} + \excess{\z_i} - \excess{\y_i}.
    \end{align*}
    Substituting this into \cref{eq:delta_r_bound} causes many terms to cancel,
    leaving
    \begin{align*}
      \Delta_R(\y_i-, \z_i+)
      \leq \Delta_{\log}(\y_i, \z_i)
        + \tail{\z_i} \frac{
            \excess{\z_i} - \excess{\y_i}}{%
            \coload{\y_i} \cdot \coload{\z_i}}
    \end{align*}
    which combined with \cref{tab:monotonicity, eq:delta_q_bound}
    implies~\cref{eq:delta_valley}, as desired.
    \qedhere
  \end{step}
\end{proof}

\section{Additional Implications of \mserpt{}'s Approximation Ratio}
\label{sec:implications}

In this section we discuss additional implications of
the fact that \mserpt{} is a constant-factor approximation of Gittins,
resolving two open questions in M/G/1 scheduling theory.
\Cref{sub:fb_imrl} addresses the performance of FB
for job size distributions with the
\emph{increasing mean residual lifetime} (IMRL) property,
and \cref{sub:mlps} addresses the performance achievable
by policies in the \emph{multilevel processor sharing} (MLPS) class.

\subsection{Performance of FB for IMRL Job Size Distributions}
\label{sub:fb_imrl}

\begin{definition}
  \label{def:imrl}
  A job size distribution~$X$ has the
  (strictly) \emph{increasing mean residual lifetime} (IMRL) property
  if a job's expected remaining size $\E{X - a \given X > a}$
  is (strictly) increasing in its age~$a$.
\end{definition}

Consider the setting of an M/G/1 with an IMRL job size distribution.
In this IMRL setting, the greater a job's age,
the greater its expected remaining size.
We therefore might expect that the FB policy,
which prioritizes jobs of lower age,
would yield low mean response time.
In fact, it was believed for some time that
FB was optimal for the IMRL setting \citep{dhr_dmrl_optimality_righter}.
However, \citet{fb_nonoptimality_aalto} found a flaw in the proof,
along with a counterexample IMRL job size distribution for which
FB is not optimal.
While \citet{fb_nonoptimality_aalto} show that
FB has lower mean response time than PS in the IMRL setting,
whether FB is close to optimal for the IMRL setting
is an open question.

The following corollary resolves this question
for the case of strictly IMRL job size distributions.
It turns out that \mserpt{} and FB are equivalent in this case,
because the strictly IMRL property implies
\mserpt{}'s rank function is strictly increasing, just like FB's.
This means FB has the same approximation ratio as \mserpt{}
for strictly IMRL job size distributions.

\begin{corollary}
  For the problem of preemptive scheduling to minimize mean response time
  in an M/G/1 queue with unknown job sizes,
  if the job size distribution is strictly IMRL,
  FB is a constant-factor approximation.
\end{corollary}

\subsection{Performance Achievable by MLPS Policies}
\label{sub:mlps}

\emph{Multilevel processor sharing} (MLPS) policies
are a class of preemptive scheduling policies
introduced by \citet{book_kleinrock}.
An MLPS policy is specified by a list of threshold ages
$0 = a_0, a_1, a_2, \dots{}$,
where interval $[a_i, a_{i + 1}]$ is the \emph{$i$th level}.
Jobs with ages in lower levels have priority over those in higher levels,
and within each level,
jobs are scheduled using one of FCFS, FB, or PS.
While we know how to analyze the mean response time of any MLPS policy
\citep{mlps_analysis_kleinrock, book_kleinrock, mlps_analysis_guo},
\emph{optimizing} an MLPS policy,
meaning choosing the threshold ages and scheduling policies within each level
to minimize mean response time,
is an open problem \citep{ps_beyond_aalto, mlps_two-level_aalto}.\footnote{%
  We note that Gittins is the solution for the special case
  where all jobs are present at the start,
  because without arrivals,
  any SOAP policy, including Gittins \citep{mlps_gittins_aalto},
  acts like an MLPS policy based on its hills and valleys.}

The following corollary takes a major step towards solving this problem.
It turns out that \mserpt{} is an MLPS policy:
its levels are the hills and valleys,
with FB used within each hill and FCFS used within each valley.
While \mserpt{} is not always the optimal MLPS policy,
we know it performs within a constant factor of Gittins.

\begin{corollary}
  For any job size distribution,
  there exists an MLPS policy, namely \mserpt{},
  with mean response time a constant factor times that of Gittins.
\end{corollary}

Combining this with results on the RMLF policy \citep{rmlf_zwart}
implies the following additional corollary.\footnote{%
  RMLF resembles an MLPS policy,
  but it is not one because it uses randomization.}

\begin{corollary}
  For any job size distribution,
  there exists an MLPS policy, namely \mserpt{},
  whose mean response time ratio compared to SRPT
  is at most $O(\log(1/(1-\rho)))$ in the $\rho \to 1$ limit.
\end{corollary}

\section{Lower Bound on \mserpt's Approximation Ratio}
\label{sec:pathological}

We have shown that \mserpt{} is a $5$\=/approximation
for minimizing mean response time.
The natural followup question is: what case is worst for \mserpt{}?
We have yet to find a scenario
in which \mserpt{} performs $5$~times worse than Gittins.
Instead, the largest ratio we have observed so far is~$2$.
This occurs with the following pathological job size distribution,
where $\delta \in (0, 1)$ is small:
\begin{align*}
  X =
  \begin{cases}
    1 - \delta & \text{w.p. } 1 - \delta \\
    1 & \text{w.p. } \delta - \delta^2 \\
    \delta^{-1} + 1 & \text{w.p. } \delta^2.
  \end{cases}
\end{align*}
That is, nearly all jobs are size $1 - \delta$,
and nearly all the rest are size~$1$.

How do the \mserpt{} and Gittins rank functions differ for~$X$?
Computing ranks using \cref{def:gittins, def:mserpt}, we find
\begin{alignat*}{2}
  \rank[\mserpt]{0}
  <{}& \rank[\mserpt]{1 - \delta}
  &&< \rank[\mserpt]{1} \\
  \rank[\gittins]{1 - \delta}
  <{}& \centermathcell{\rank[\gittins]{0}}
  &&< \rank[\gittins]{1}
\end{alignat*}
In terms of hills and valleys,
both \mserpt{} and Gittins have a hill age at~$1$,
but \mserpt{} has an additional hill age at $1 - \delta$.
But \mserpt{}'s extra hill age increases mean response time:
a job of age $1 - \delta$ will almost always finish
with just~$\delta$ more work,
so it would be better to give those jobs priority over jobs at age~$0$.
Gittins does not make this mistake.

We now compute the mean response times of \mserpt{} and Gittins
for a system with job size distribution~$X$.
Suppose the load is $\rho = 1 - \epsilon$,
where $\epsilon \in (0, 1)$ is small.
We have
\begin{align*}
  \coload{0} &= 1 &
  \coload{1 - \delta} &\approx \delta + \epsilon &
  \coload{1} &\approx \delta + \epsilon &
  \coload{\infty} &= \epsilon \\
  \excess{0} &= 0 &
  \excess{1 - \delta} &\approx \tfrac{1}{2} &
  \excess{1} &\approx \tfrac{1}{2} &
  \excess{\infty} &\approx 1,
\end{align*}
where the approximations assume $\delta, \epsilon \ll 1$.
By \cref{prop:waiting_hills_valleys, prop:residence_hills_valleys},
the mean response time of \mserpt{} is
\begin{align*}
  \E{\response[\mserpt]{}}
  &\approx \E{\response[\mserpt]{1 - \delta}}
    + \delta\E{\response[\mserpt]{1}}
    + \delta^2\E{\response[\mserpt]{\delta^{-1} + 1}} \\
  &= \gp[\bigg]{
      \frac{\excess{1 - \delta}}{\coload{0} \cdot \coload{1 - \delta}}
      + \frac{1 - \delta}{\coload{0}}
    } + \delta\gp[\bigg]{
      \frac{\excess{1}}{\coload{1 - \delta} \cdot \coload{1}}
      + \frac{1}{\coload{1 - \delta}}
    } \iftwocol{\\ &\qquad}{} + \delta^2\gp[\bigg]{
      \frac{\excess{\infty}}{\coload{1} \cdot \coload{\infty}}
      + \frac{\delta^{-1} + 1}{\coload{1}}
    } \\
  &\approx \gp[\bigg]{
      \frac{\tfrac{1}{2}}{\delta + \epsilon}
      + 1 - \delta
    } + \delta\gp[\bigg]{
      \frac{\tfrac{1}{2}}{(\delta + \epsilon)^2}
      + \frac{1}{\delta + \epsilon}
    } \iftwocol{\\ &\qquad}{} + \delta^2\gp[\bigg]{
      \frac{1}{\epsilon \cdot (\delta + \epsilon)}
      + \frac{\delta^{-1} + 1}{\delta + \epsilon}
    } \\
  &\approx \frac{1}{2(\delta + \epsilon)}\gp[\bigg]{
      \frac{2\delta + \epsilon}{\delta + \epsilon}
      + \frac{2\delta^2}{\epsilon}
    }.
\end{align*}
We now analyze the mean response time of Gittins.
One can show using the full SOAP analysis \citep{soap_scully}
that when $\delta, \epsilon \ll 1$,
\cref{prop:waiting_hills_valleys, prop:residence_hills_valleys}
give approximate equalities for Gittins, so
\begin{align*}
  \E{\response[\gittins]{}}
  &\approx \E{\response[\gittins]{1 - \delta}}
    + \delta\E{\response[\gittins]{1}}
    + \delta^2\E{\response[\gittins]{\delta^{-1} + 1}} \\
  &\approx \gp[\bigg]{
      \frac{\excess{1}}{\coload{0} \cdot \coload{1}}
      + 1 - \delta
    } + \delta\gp[\bigg]{
      \frac{\excess{1}}{\coload{0} \cdot \coload{1}}
      + 1
    } \iftwocol{\\ &\qquad}{} + \delta^2\gp[\bigg]{
      \frac{\excess{\infty}}{\coload{1} \cdot \coload{\infty}}
      + \delta^{-1} + 1
    } \\
  &\approx \gp[\bigg]{
      \frac{\tfrac{1}{2}}{\delta + \epsilon}
      + 1 - \delta
    } + \delta\gp[\bigg]{
      \frac{\tfrac{1}{2}}{\delta + \epsilon}
      + 1
    } + \delta^2\gp[\bigg]{
      \frac{1}{\epsilon \cdot (\delta + \epsilon)}
      + \delta^{-1} + 1
    } \\
  &\approx \frac{1}{2(\delta + \epsilon)}\gp[\bigg]{
      1 + \frac{2\delta^2}{\epsilon}
    }.
\end{align*}
This makes the mean response time ratio approximately
\begin{align*}
  \frac{\E{\response[\mserpt]{}}}{\E{\response[\gittins]{}}}
  \approx \frac{
      2\delta^3 + 2\delta\epsilon + \epsilon^2
    }{
      2\delta^3 + \delta\epsilon + \epsilon^2
    }.
\end{align*}
This ratio is at most~$2$,
and it can approach $2$ in any limit where the $\delta\epsilon$ term dominates.
This happens if we set $\epsilon = \delta^{3/2}$ in the $\delta \to 0$ limit,
so \mserpt{}'s approximation ratio is at least~$2$.

\section{Why Closing the Gap is Hard}
\label{sec:tightening}

In preliminary numerical studies,
omitted for lack of space,
we have computed the mean response time ratio
between \mserpt{} and Gittins for a variety of job size distributions.
We have yet to observe a ratio greater than~$2$,
with \cref{sec:pathological} describing the worst case we have found,
motivating the following conjecture.

\begin{conjecture}
  \label{conj:approximation_ratio}
  For the problem of preemptive scheduling to minimize mean response time
  in an M/G/1 queue with unknown job sizes,
  the approximation ratio of \mserpt{} is~$2$.
\end{conjecture}

The lower bound of~$2$ on \mserpt{}'s approximation ratio is less than
the upper bound of~$5$ from \cref{thm:approximation_ratio}.
What would it take to close the gap?
Recall from \cref{fig:outline} that we prove \cref{thm:approximation_ratio}
by combining the four following bounds.
The main obstacle to closing the gap is that
\emph{each of the four bounds is tight in some setting}.
\begin{enumerate}[(i)]
\item
  \label{item:waiting}
  \Cref{lem:waiting} gives an upper bound on
  $\E{\waiting[\mserpt]{}}/\E{\waiting[\gittins]{}}$.
  \begin{itemize}[\labelitemii]
  \item
    It is tight for the scenario described in
    \cref{sec:pathological}.
  \end{itemize}
\item
  \label{item:residence}
  \Cref{lem:residence} gives an upper bound on $\E{\residence[\mserpt]{}}$.
  \begin{itemize}[\labelitemii]
  \item
    It is tight in the $\rho \to 1$ limit
    for Pareto job size distributions
    with shape parameter $\alpha \approx 1$ \citep{fb_heavy_zwart}.
  \end{itemize}
\item
  \label{item:residence_size}
  \Cref{prop:residence_size} gives a lower bound on
  $\E{\residence[\gittins]{}}$.
  \begin{itemize}[\labelitemii]
  \item
    It is tight when Gittins is equivalent to FCFS,
    which occurs for some job size distributions
    \citep{m/g/1_gittins_aalto}.
  \end{itemize}
\item
  \label{item:response_srpt}
  \Cref{prop:response_srpt} gives a lower bound on $\E{\response[\gittins]{}}$.
  \begin{itemize}[\labelitemii]
  \item
    It is tight in the $\rho \to 1$ limit
    for Pareto job size distributions
    with shape parameter $\alpha \approx 1$ \citep{fb_heavy_zwart}.
  \end{itemize}
\end{enumerate}
The fact that each bound is tight
means that tightening \cref{thm:approximation_ratio}
requires new insight.

\newcommand{\bounds}{%
  bounds~\cref{%
    item:waiting,
    item:residence,
    item:residence_size,
    item:response_srpt}}

Although \bounds{} are all tight,
they are tight in \emph{different settings},
meaning for different loads~$\rho$ and job size distributions~$X$.
This hints at a possible approach
to tightening \cref{thm:approximation_ratio}:
we could refine \bounds{} in a way that makes them
more sensitive to the setting, especially the job size distribution.
As an example of what this might mean,
the settings in which bounds~\cref{item:waiting, item:residence_size} are tight
have $\Var{X^2} < \infty$,
while those in which bounds~\cref{item:residence, item:response_srpt} are tight
have $\Var{X^2} = \infty$.
Thus, we might be able to improve on \cref{thm:approximation_ratio}
if we refine each of \bounds{}
by ``conditioning'', meaning splitting into cases,
on whether $\Var{X}$ is finite.

With that said, we suspect that refining \bounds{}
is more involved than simply conditioning on whether $\Var{x} = \infty$.
In the rest of this section we review each bound,
explain the settings in which they are tight in more detail,
and discuss opportunities for refining or replacing them.

\subsection{Tightening the \mserpt{} Upper Bounds}

We begin with bound~\cref{item:waiting}, \cref{lem:waiting},
which implies $\E{\waiting[\mserpt]{}} \leq 2\E{\waiting[\gittins]{}}$.
This bound is tight for the scenario described in \cref{sec:pathological}.
To find opportunities for tightening,
recall that the proof of \cref{lem:waiting} works
by looking at one valley at a time,
showing a ratio bound for each valley separately.
When proving the bound for valley $(u, v]$,
we use the fact that $\coload{u} \leq 1$,\footnote{%
  Specifically, we apply \cref{lem:coload} with $a = u$,
  and \cref{lem:coload}'s proof uses $\coload{a} \leq 1$.}
but this is tight for at most one valley.
In the job size distribution from \cref{sec:pathological},
nearly every job's size is in a valley with $\coload{u} = 1$,
which is why \cref{lem:waiting} is tight in that scenario.
But many job size distributions do not have nearly all job sizes in one valley.
We could perhaps refine \cref{lem:waiting}
by conditioning on a parameter related to valleys,
such as a bound $\zeta \in [0, 1]$ such that
$\P{X \in (u, v]} \leq \zeta$ for all valleys $(u, v]$.

Bound~\cref{item:residence}, \cref{lem:residence}, says
$\E{\residence[\mserpt]{}} \leq \E{\waiting[\mserpt]{}} + \ell_\rho$,
where
\begin{align*}
  \ell_\rho = \gp[\bigg]{\frac{1}{\rho}\log\frac{1}{1 - \rho}}\E{X}.
\end{align*}
\Cref{lem:residence} can be tight in the $\rho \to 1$ limit
when $X$ has a Pareto job size distribution.
For shape parameter $\alpha \in (1, 2)$,
if $\tail{x} = (1 + x)^{-\alpha}$,
a result of \citet[Section~4.2.1]{fb_heavy_zwart} implies\footnote{%
  \Citet{fb_heavy_zwart} consider the FB policy,
  but \mserpt{} and FB are equivalent for this job size distribution
  because it has the IMRL property (\cref{def:imrl}).}
\begin{align}[c]
  \label{eq:pareto_heavy}
  \E{\waiting[\mserpt]{}}
  &\approx \frac{\alpha(\alpha - 1)}{2 - \alpha} \cdot \ell_\rho \\
  \E{\residence[\mserpt]{}}
  &\approx \alpha \cdot \ell_\rho
\end{align}
as $\rho \to 1$.
This means the tightness of \cref{lem:residence} in the $\rho \to 1$ limit
depends on~$\alpha$:
it is tight for $\alpha \approx 1$ but extremely loose for $\alpha \approx 2$.
Similar reasoning shows the bound is also loose for $\alpha > 2$
\citep[Section~4.1.1]{fb_heavy_zwart}.
This suggests that we could try to refine \cref{lem:residence}
by conditioning on the tail behavior of~$X$.
A concrete opportunity for tightening is in
\cref{step:delta_hill} of the proof:
the difference between the two sides of the final inequality is
$2 \lambda \tail{x}^2 \cdot \excess{x} / \coload{x}^3$,
whose contribution is negligible for $\alpha \approx 1$
but dominates for larger~$\alpha$ \citep{fb_heavy_zwart}.
\Cref{step:delta_valley} of the proof has a similar opportunity,
but the difference term is more complicated.
Another obstacle to this approach
is the lack of results in the style of \citet{fb_heavy_zwart}
that hold outside the $\rho \to 1$ limit.

\subsection{Tightening the Gittins Lower Bounds}
\label{sub:tightening_lower}

Bound~\cref{item:residence_size}, \cref{prop:residence_size},
gives a trivial lower bound
on Gittins's mean residence time,
namely $\E{\residence[\gittins]{}} \geq \E{X}$.
But even this trivial bound is tight for some job size distributions,
namely those with the
\emph{new better than used in expectation} property
\citep{m/g/1_gittins_aalto}.
This is because the Gittins policy is equivalent to FCFS for such distributions
\citep{m/g/1_gittins_aalto},
and FCFS has mean residence time $\E{X}$.
However, a result of \citet[Proposition~9]{mlps_gittins_aalto} implies that
if Gittins is equivalent to FCFS for some distribution~$X$,
then \mserpt{} is also equivalent to FCFS.
That is, when Gittins has very low residence time, so does \mserpt{}.
This hints that what we would really like is
a direct bound on $\E{\residence[\mserpt]{}}/\E{\residence[\gittins]{}}$.
Unfortunately,
the residence time formula in \cref{prop:residence_hills_valleys}
gives an upper bound on $\E{\residence[\gittins]{}}$,
whereas we need a lower bound.
Even if we could bound the gap between
$\E{\residence[\gittins]{}}$
and the upper bound in \cref{prop:residence_hills_valleys},
bounding $\E{\residence[\mserpt]{}}/\E{\residence[\gittins]{}}$
would likely still be at least as challenging as proving \cref{lem:waiting}.

We finally turn to bound~\cref{item:response_srpt}, \cref{prop:response_srpt},
which is a corollary of a result of
\citet[Theorem~5.8]{smart_insensitive_wierman}.
It says $\E{T} \geq \ell_\rho$ for \emph{any} scheduling policy,
including size-based policies like SRPT.
Despite this, by \cref{eq:pareto_heavy},
\cref{prop:response_srpt} is tight in the $\rho \to 1$ limit
when $X$ has a Pareto job size distribution
with shape parameter $\alpha \approx 1$.
We are not aware of any other simple lower bound on SRPT's mean response time
that holds for all job size distributions.
One possibility for refining the bound would be to
parametrize them along similar lines as further results of
\citet[Theorems~5.4, 5.7, and~5.9]{smart_insensitive_wierman}.
Of course, we would prefer a bound that holds
only for policies that, like Gittins, do not use job size information,
but we suspect such a result requires new techniques.

\section{Conclusion}

We introduce \mserpt{},
the first non-Gittins policy
proven to achieve mean response time within a constant factor of Gittins's.
Specifically, we show that \mserpt{} is a $5$\=/approximation of Gittins,
with an even smaller approximation ratio at lower loads
(\cref{thm:approximation_ratio}).
In addition to being an important result in its own right,
the fact that \mserpt{} has near-optimal mean response time
resolves two open questions in M/G/1 scheduling theory
(\cref{sec:implications}).

An open question is whether \mserpt{}'s approximation ratio is less than~$5$.
We conjecture that the true approximation ratio is~$2$
(\cref{conj:approximation_ratio}).
Another open question is how SERPT's mean response time
compares to \mserpt{}'s.
In preliminary numerical studies,
we have observed very similar performance from SERPT and \mserpt{},
with each sometimes outperforming the other,
so we conjecture that SERPT is also a constant-factor approximation of Gittins.

\begin{acks}
  This work was supported by NSF-CSR-1763701, NSF-XPS-1629444 and
  a Microsoft Faculty Award 2018.
  Ziv Scully was supported by the NSF GRFP
  under grants DGE-1745016 and DGE-125222
  and an ARCS Foundation scholarship.
  We thank the anonymous referees for their helpful comments.
\end{acks}

\bibliographystyle{ACM-Reference-Format}
\bibliography{refs}

\appendix

\section{No Approximation Ratio for Traditional Policies}
\label{app:infinite_ratio}

In this appendix we discuss the performance
of three traditional policies: FCFS, FB, and PS.
We will show that none of these policies are constant-factor approximations
for mean response time.
That is, the ratio of each policy's mean response times to that of Gittins
can be unboundedly large.

FCFS has mean response time \citep{book_harchol-balter}
\begin{align*}
  \E{\response[\policyname{FCFS}]{}}
  = \frac{\lambda\E{X^2}}{2(1 - \rho)} + \E{X}.
\end{align*}
This is infinite if $X$ has infinite variance,
but other policies have finite mean response time
for all job size distributions,
so FCFS has no constant-factor approximation ratio.

For the specific case where all jobs have size~$x$,
FB has mean response time \citep{book_harchol-balter}
\begin{align*}
  \E{\response[\policyname{FB}]{}}
  = \frac{\lambda x^2}{2(1 - \rho)^2} + \frac{x}{1 - \rho}.
\end{align*}
This is worse than FCFS's mean response time in the same case
by a factor of $1/(1 - \rho)$,
which becomes arbitrarily large in the $\rho \to 1$ limit,
so FB has no constant-factor approximation ratio.

PS has mean response time \citep{book_harchol-balter}
\begin{align*}
  \E{\response[\policyname{PS}]{}} = \frac{\E{X}}{1 - \rho}.
\end{align*}
That is, the response time of PS is insensitive
to the details of the job size distribution,
depending only on the mean.
While PS is thus generally considered to have reasonable performance
for all job size distributions,
there are certain distributions where other policies outperform PS by far.
For example, \citet{fb_heavy_zwart} show that
when $X$ is a Pareto distribution with shape parameter $\alpha \in (1, 2)$,
FB has mean response time that scales as
\begin{align*}
  \E{\response[\policyname{FB}]{}}
  \approx \frac{\alpha \E{X}}{2 - \alpha}\log\frac{1}{1 - \rho}
\end{align*}
in the $\rho \to 1$ limit.
Thus, the mean response time ratio between PS and FB
becomes arbitrarily large in the $\rho \to 1$ limit,
so PS has no constant-factor approximation ratio.

\section{Difficulty of Computing the Gittins Policy}
\label{app:gittins_hard}

In this appendix we discuss in more detail
why it is difficult to compute the Gittins rank function.
We begin with the simpler case of discrete job size distributions
(\cref{sub:gittins_hard_discrete})
before turning to continuous job size distributions
(\cref{sub:gittins_hard_continuous}).

\subsection{Discrete Job Size Distributions}
\label{sub:gittins_hard_discrete}

All the algorithms discussed in this section
assume input in the form of a list of $(x, p)$ pairs sorted by~$x$,
where $x$ is a support point and $p$ is the probability of outcome~$x$.

The problem of computing the Gittins rank\footnote{%
  Most literature refers to the \emph{Gittins index},
  which is simply the reciprocal of the Gittins rank.}
of all states in a finite Markov chains
is a well studied problem
for which the best known algorithms take $O(n^3)$ time,
where $n$ is the number of states in the Markov chain
\citep{gittins_index_computation_chakravorty}.
The reader may recall that we claim in \cref{tab:complexity}
that Gittins takes $O(n^2)$ time to compute.
This is due to two discrepancies between algorithms in the literature
and the problem we consider,
namely computing the Gittins rank function
for a discrete job size distribution.
\begin{itemize}
\item
  Algorithms in the literature assume an arbitrary finite Markov chain.
  However, a discrete job size distribution has a very simple structure
  when viewed as a Markov chain.
  Each support point is a state, and each has only two transitions
  with nonzero probability:
  to the next support point and to a terminal state.\footnote{%
    The terminal state is the maximum support point.
    Additionally, there is an initial state at age~$0$.
    In the following discussion, any mention of ``adjacent support points''
    also applies to the interval between $0$ and the first support point.}
  In this respect, our problem is \emph{easier}
  than the one solved in the literature.
\item
  Algorithms in the literature compute the Gittins rank at each state,
  which in our case corresponds to each support point.
  However, the full Gittins rank function assigns ranks to all ages,
  and ages between adjacent support points are not covered
  by algorithms in the literature.
  In this respect, our problem is \emph{harder}
  than the one solved in the literature.
\end{itemize}

It turns out that the former difference has the greater impact.
Specifically, if one uses sparse matrix operations,
algorithms in the literature can be implemented such that
they take only $O(n^2)$ time \citep{gittins_index_computation_chakravorty},
because the Markov chain of a discrete job size distribution has only
$O(n)$ transitions with nonzero probability.
The output of this algorithm is the Gittins rank of each support point,
but it remains to compute the rank function at other ages.
Between each pair of adjacent support points,
the Gittins rank function is piecewise linear with at most $O(n)$ segments.
This means a post-processing step taking $O(n)$ time per support point,
and thus $O(n^2)$ time total,
can fill in the gaps between adjacent support points.

We have summarized how to use state-of-the-art algorithms from the literature
to compute the Gittins rank function in $O(n^2)$ time.
Whether there exists an algorithm computing the Gittins rank function
in $o(n^2)$ time remains an open problem.

Finally, we briefly sketch an algorithm
that computes the SERPT and \mserpt{} rank functions in $O(n)$ time.
Computing $\rank[\serpt]{x}$ at each support point~$x$
can be done with a table containing
$\tail{x}$ and $\int_x^\infty \tail{t}\d{t}$ for each support point~$x$,
which can be generated with scans that take $O(n)$ time each.
This yields the SERPT rank at each support point,
and an additional $O(n)$ scan yields the same for \mserpt{}.
Between adjacent support points,
SERPT's rank function simply decreases at slope~$1$
while \mserpt{}'s is constant.

\subsection{Continuous Job Size Distributions}
\label{sub:gittins_hard_continuous}

The Gittins policy for continuous job size distributions
has received some attention,
with results characterizing the Gittins rank function available
under various assumptions on the job size distribution
\citep{m/g/1_gittins_aalto, mlps_gittins_aalto}.
However, none of the prior work explicitly addresses
computing the Gittins policy for a general continuous job size distribution.
Here we review the most general characterization result
and show why it does not solve
the problem of computing the Gittins rank function.

\Citet[Propositions~1 and~11]{mlps_gittins_aalto} show the following result.
Suppose there exist ages $0 = v_0, u_1, v_1, u_2, v_2, \dots$ such that
for all $i \geq 1$, the job size distribution's hazard rate~$\hazard{}$ is
\begin{itemize}
\item
  strictly decreasing for $(u_i, v_i)$ and
\item
  increasing for $(v_{i - 1}, u_i)$.
\end{itemize}
Then for all $i \geq 1$, there exists an age $w_i \in [u_i, v_i]$ such that
the Gittins rank function~$\rank[\gittins]{}$ is
\begin{itemize}
\item
  strictly increasing for $(u_i, v_i)$ and
\item
  decreasing for $(w_{i - 1}, u_i)$.\footnote{%
    We define $w_0 = 0$.}
\end{itemize}
Knowing something about the monotonicity of the Gittins rank function
is potentially helpful for computing it.
However, the results of \citet{mlps_gittins_aalto}
do not provide a way to \emph{compute} the critical ages~$w_i$.
Moreover, even if we could compute the ages $w_i$,
as we explain below,
computing the rank function can be at least as hard as in the discrete case.

For each age~$a$, there is an optimal stopping age $b_*(a)$
that solves the optimization problem in $\rank[\gittins]{a}$
(\cref{def:gittins}).
We know by results of \citet{mlps_gittins_aalto}
that if $b_*(a) > a$,
then $b_*(a)$ lies in interval $[u_i, w_i]$ for some~$i$,
but we do not know which~$i$.
This makes the search for $b_*(a)$
intractable if there are infinitely many intervals $[u_i, w_i]$
and at least as hard as the discrete case if there are finitely many.

\section{SOAP Mean Response Time Using Hills and Valleys}
\label{app:soap_hills_valleys}

\Cref{prop:waiting_hills_valleys, prop:residence_hills_valleys}
follow immediately from results of
\citet[Theorem~5.5, see also Lemmas~5.2 and~5.3]{soap_scully}.
The main obstacle is a difference in notation.
Below we translate from the notation in our paper
to the notation of \citet{soap_scully}:
\begin{align*}
  \coload{\y(x)}
  &= 1 - \rho^{\text{new}}[r^{\text{worst}}_x(0)]
  \geq 1 - \rho^{\text{new}}[r^{\text{worst}}_x(a)] \\
  \coload{\z(x)} &= 1 - \rho^{\text{old}}_0[r^{\text{worst}}_x(0)] \\
  \excess{\z(x)}
  &= \frac{\lambda}{2}\E{X^{\text{old}}_0[r^{\text{worst}}_x(0)]}
  \leq \frac{\lambda}{2}\sum_{i = 0}^\infty
    \E{X^{\text{old}}_i[r^{\text{worst}}_x(0)]}.
\end{align*}
When the rank function is monotonic,
showing that the bounds in
\cref{prop:waiting_hills_valleys, prop:residence_hills_valleys}
become equalities
boils down to proving that the two inequalities above become equalities.
We first note that any decreasing rank function is equivalent to FCFS.
But FCFS can also be expressed by a constant rank function,
which is weakly increasing.
We therefore restrict our attention to increasing rank functions,
for which the following properties are easily shown:
\begin{itemize}
\item
  $r^{\text{worst}}_x(a) = r^{\text{worst}}_x(0)$ for all ages~$a$
  \citep[Definition~4.1]{soap_scully}, and
\item
  $X^{\text{old}}_i[r] = 0$ with probability~$1$
  for all ranks~$r$ and integers $i \geq 1$.
  \citep[Definition~4.3]{soap_scully}.
\end{itemize}
Thus, both inequalities above become equalities for monotonic rank functions.

\section{Deferred Proofs}

\thmapproximationratio

\begin{proof}
  \label{pf:approximation_ratio}
  Bounding mean response time amounts to bounding
  mean waiting and residence times.
  By \cref{lem:waiting},
  \begin{align*}
    \E{\waiting[\mserpt]{}}
    \leq \frac{2}{1 + \sqrt{1 - \rho}} \E{\waiting[\gittins]{}},
  \end{align*}
  and by \cref{lem:residence},
  \begin{align}
    \label{eq:residence_bound}
    \E{\residence[\mserpt]{}}
    \leq \E{\waiting[\mserpt]{}}
      + \gp[\bigg]{\frac{1}{\rho}\log\frac{1}{1 - \rho}}\E{X}.
  \end{align}
  We can give two different bounds on
  the last term of~\cref{eq:residence_bound},
  each of which yields a bound on the mean response time ratio.
  Applying \cref{prop:residence_size} yields
  \begin{align*}
    \frac{\E{\response[\mserpt]{}}}{\E{\response[\gittins]{}}}
    &\leq \frac{%
      \frac{4}{1 + \sqrt{1 - \rho}} \E{\waiting[\gittins]{}}
      + \gp[\Big]{\frac{1}{\rho}\log\frac{1}{1 - \rho}}
        \E{\residence[\gittins]{}}}{%
      \E{\waiting[\gittins]{}}
      + \E{\residence[\gittins]{}}} \\
    &\leq \max\curlgp[\bigg]{%
      \frac{4}{1 + \sqrt{1 - \rho}},
      \frac{1}{\rho}\log\frac{1}{1 - \rho}}.
  \end{align*}
  Applying \cref{prop:response_srpt} instead yields
  \begin{align*}
    \frac{\E{\response[\mserpt]{}}}{\E{\response[\gittins]{}}}
    \leq \frac{%
      \gp[\Big]{1 + \frac{4}{1 + \sqrt{1 - \rho}}}\E{\waiting[\gittins]{}}
      + \E{\residence[\gittins]{}}}{%
      \E{\waiting[\gittins]{}}
      + \E{\residence[\gittins]{}}}
    \leq 1 + \frac{4}{1 + \sqrt{1 - \rho}}.
  \end{align*}
  Taking the minimum of these two bounds gives us
  \begin{align*}
    \frac{\E{\response[\mserpt]{}}}{\E{\response[\gittins]{}}}
    \leq \min\curlgp[\bigg]{\max\curlgp[\bigg]{%
        \frac{4}{1 + \sqrt{1 - \rho}},
        \frac{1}{\rho}\log\frac{1}{1 - \rho}},
      1 + \frac{4}{1 + \sqrt{1 - \rho}}},
  \end{align*}
  which expands to the desired piecewise bound.
\end{proof}

\lemhillsubset*

\begin{proof}
  \label{pf:hill_subset}
  We prove the result for hill ages.
  The corresponding result for hill sizes then follows immediately
  from the observation that $x$ is a hill size
  if and only if there exists $\epsilon > 0$ such that
  all ages in $[x, x + \epsilon)$ are hill ages,
  so we can simply apply the hill age result to those intervals.

  It is immediate from \cref{def:mserpt}
  that SERPT and \mserpt{} have the same hill ages,
  so in this proof, we work with SERPT instead of \mserpt{}.

  At the core of our argument is the following definition.
  For ages $a < b$, let
  \begin{align*}
    \eta(a, b)
    &= \frac{\int_a^b \tail{t} \d{t}}{\tail{a} - \tail{b}} \\
    \eta(a, a)
    &= \lim_{b \to a} \eta(a, b)
    = \frac{\tail{a}}{\density{a}}
    = \frac{1}{\hazard{a}} \\
    \eta(a, \infty)
    &= \lim_{b \to \infty} \eta(a, b)
    = \E{X - a \given X > a}.
  \end{align*}
  The function~$\eta$ is a version of the \emph{efficiency function}
  commonly used in the M/G/1 Gittins policy literature
  \citep{m/g/1_gittins_aalto, mlps_gittins_aalto}.
  Its continuity is inherited from the fact that $X$ has a density function
  (\cref{sec:model}).
  It is closely related to the rank functions of SERPT and Gittins:\footnote{%
    The minimum in $\rank[\gittins]{a}$ always exists because
    we allow $b = a$ and $b = \infty$.}
  \begin{align}
    \label{eq:rank_cmp}
    \rank[\serpt]{a} &= \eta(a, \infty) \\
    \rank[\gittins]{a} &= \min_{b \geq a} \eta(a, b) \leq \rank[\serpt]{a}
  \end{align}
  It is simple to verify that for any ages $a \leq b \leq c$,
  \begin{alignat}[c]{3}
    \label{eq:score_cmp}
    \eta(a, b) &\leq{}& \eta(a, c) &\leq{}& \eta(b, c) \\
    && \centermathcell{\Updownarrow} && \\
    \eta(a, b) &\leq{}& \eta(a, c) && \\
    && \centermathcell{\Updownarrow} && \\
    && \eta(a, c) &\leq{}& \eta(b, c) \\
    && \centermathcell{\Updownarrow} && \\
    \eta(a, b) &&\centermathcell{{}\leq{}} && \eta(b, c).
  \end{alignat}
  and similarly for strict inequalities when $a < b < c$.

  A useful intuition is that
  $\eta(a, b)$ gives a ``score'' to the interval $[a, b]$,
  where lower scores are better.
  SERPT gives a job at age~$a$ rank equal to the score of $[a, \infty]$,
  while Gittins is pickier,
  choosing the best score among all intervals that start at~$a$.
  What \cref{eq:score_cmp} says is that
  if we divide an interval into two pieces,
  the score of the interval is between scores of its pieces.

  Let $v$ be a Gittins hill age and consider any age $u < v$.
  We want to show that $v$ is a SERPT hill age,
  which amounts to showing $\rank[\serpt]{u} < \rank[\serpt]{v}$.
  By \cref{eq:rank_cmp, eq:score_cmp}, it suffices to show
  \begin{align}
    \label{eq:score_goal}
    \eta(u, v) \leq \rank[\gittins]{v}.
  \end{align}
  For simplicity, we show \cref{eq:score_goal} only for $u = 0$,
  explaining at the end of the proof why we can do so
  without loss of generality.

  To show \cref{eq:score_goal} with $u = 0$,
  we need to understand $\eta(0, v)$.
  We can partition $[0, v]$ into Gittins hills and valleys,
  meaning there exist
  \begin{align*}
    0 = \z_0 \leq \y_1 < \z_1 < \dots
    < \y_n < \z_n \leq \y_{n + 1} = v
  \end{align*}
  such that
  \begin{itemize}
    \item
      $(\y_i, \z_i]$ is a Gittins valley for all $i \in \{1, \dots, n\}$,
    \item
      $(\z_i, \y_{i + 1}]$ is a Gittins hill
      for all $i \in \{1, \dots, n\}$, and
    \item
      either $\z_0 = \y_1$ or $(\z_0, \y_1]$ is a Gittins hill.
  \end{itemize}
  By repeatedly applying \cref{eq:score_cmp},
  it suffices to show that for each hill $(\z_i, \y_{i + 1}]$,
  \begin{align}
    \label{eq:score_hill}
    \eta(z_i, y_{i + 1}) < \rank[\gittins]{v},
  \end{align}
  and that for each valley $(\y_i, \z_i]$,
  \begin{align}
    \label{eq:score_valley}
    \eta(y_i, z_i) < \rank[\gittins]{v}.
  \end{align}
  We prove these bounds
  in \cref{step:score_hill, step:score_valley} below, respectively.

  \setcounter{step}{0}
  \begin{step}[bound for Gittins hills]
    \label{step:score_hill}
    Let $(\z_i, \y_{i + 1}]$ be a Gittins hill.
    Continuity of $\rank[\gittins]{}$ (\cref{sub:policies})
    and a result of \citet[Proposition~3]{mlps_gittins_aalto}
    together imply that for all $a \in [\z_i, \y_{i + 1})$,
    \begin{align*}
      \rank[\gittins]{a} = \frac{1}{\hazard{a}}
      < \frac{1}{\hazard{\y_{i + 1}}}
      = \rank[\gittins]{\y_{i + 1}},
    \end{align*}
    from which another result \citep[Lemma~5]{mlps_gittins_aalto} yields
    \begin{align*}
      \rank[\gittins]{\z_i} = \eta(\z_i, \z_i) \leq \eta(\z_i, \y_{i + 1}).
    \end{align*}
    By \cref{eq:score_cmp}, we also have
    \begin{align*}
      \eta(\z_i, \y_{i + 1})
      \leq \eta(\y_{i + 1}, \y_{i + 1})
      = \rank[\gittins]{\y_{i + 1}}
    \end{align*}
    Combining this with the fact that $v > \y_{i + 1}$ is a Gittins hill age
    implies~\cref{eq:score_hill}, as desired.
  \end{step}

  \begin{step}[bound for Gittins valleys]
    \label{step:score_valley}
    Let $(\y_i, \z_i]$ be a Gittins valley.
    A fundamental property of Gittins
    \citep[Lemma~2.2]{book_gittins}
    implies\footnote{%
      \Citet{book_gittins} focus on a discrete setting,
      but essentially the same proof holds in our continuous setting.}
    \begin{align*}
      \rank[\gittins]{\y_i} = \eta(\y_i, \z_i).
    \end{align*}
    Combining this with the fact that $v > \y_i$ is a Gittins hill age
    implies~\cref{eq:score_valley}, as desired.
  \end{step}

  With \cref{step:score_hill, step:score_valley}
  we have shown \cref{eq:score_goal} for $u = 0$.
  To generalize the argument to $u > 0$,
  we observe that
  the rank functions of SERPT and Gittins at ages $u$ and later
  do not depend on ages earlier than~$u$.
  Consider a modified job size distribution $X' = (X - u \given X > u)$.
  Writing $\rankp{}$ for rank functions with distribution~$X'$,
  we have
  \begin{align*}
    \rankp[\serpt]{a} &= \rank[\serpt]{a + u} \\
    \rankp[\gittins]{a} &= \rank[\gittins]{a + u}
  \end{align*}
  for all ages~$a$.
  Switching job size distributions from $X$ to~$X'$
  simply shifts the rank functions by~$u$,
  so $v - u$ is a Gittins hill age for~$X'$.
  This transforms the $u > 0$ case for~$X$ into the $u = 0$ case for~$X'$.
\end{proof}

\begin{anononly}

\section{Responses to Reviewer Comments}
\label{app:rebuttal}

We begin by addressing the three main points we were asked to address
for this one-shot revision.
We then address the individual reviewer comments
that are not covered by the three main points.
We have updated references in the reviewer comments
to match the current version of the paper.

\subsection{Main Points for One-Shot Revision}
\label{sub:main_points_one-shot}

\begin{enumerate}[(1)]
  \item
    The reviewers asked us to better substantiate the claim
    that Gittins is hard to compute.
    \begin{itemize}[\labelitemii]
    \item
      We address this in \cref{app:gittins_hard}.
      Briefly, while the work of
      \citet{m/g/1_gittins_aalto, mlps_gittins_aalto}
      provides some characterization of
      what the Gittins rank function looks like,
      it does not completely solve the problem of computing it.
    \item
      We also include a brief discussion of
      what advances would be necessary
      to make the computation of Gittins more tractable
      (\cref{sub:gittins_easier}).
    \end{itemize}
  \item
    The reviewers asked whether we can close the gap between
    our bounds on \mserpt{}'s approximation ratio,
    which we show is at least~$2$ (\cref{sec:pathological})
    and at most~$5$ (\cref{thm:approximation_ratio}).
    \begin{itemize}[\labelitemii]
    \item
      We first reemphasize that our result that \mserpt{}
      is a $5$\=/approximation for mean response time is
      the first constant-factor approximation ratio
      proven for any policy other than Gittins.
      In fact, we can say something stronger.
      Suppose one used a set of traditional blind policies from the literature
      (FCFS, FB, PS, RMLF, etc.)
      to create a hybrid policy that
      for each job size distribution~$X$
      used whichever policy had the best mean response time for~$X$.
      There is no proof or combination of proofs in the literature
      showing that the resulting hybrid policy
      is a constant-factor approximation of Gittins.
    \item
      We spent significant effort attempting
      to improve our bound on the approximation ratio,
      including trying Reviewer~A's suggestion of using the connection between
      Gittins and MLPS policies,
      but we were not able to.
      The result of our effort is \cref{sec:tightening},
      which discusses in detail what obstacles there are to closing the gap
      and what approaches future work could take
      to try to overcome those obstacles.
    \end{itemize}
  \item
    The reviewers asked us to make the proofs easier to follow.
    \begin{itemize}[\labelitemii]
    \item
      We have restructured the whole paper to make it easier to follow.
      The most important change for the technical presentation
      is that \cref{sec:outline, sec:hills_valleys}
      have been completely rewritten.
      \Cref{sec:outline} now gives a high-level outline of the entire paper,
      including detailed intuitive overviews of the strategies behind
      the proofs of \cref{lem:waiting, lem:residence}.
      This high-level discussion sets the stage for \cref{sec:hills_valleys}
      to give more rigorous definitions of hills, valleys, and related concepts
      than in the previous version of the paper.
    \end{itemize}
\end{enumerate}

\subsection{Reviewer~A}
\label{sub:reviewer_a}

\begin{itemize}
\item
  \textit{Figure 3.1. When the rank is flat, M-SERPT implements FCFS. However, PS could also be a viable option. If there is any reason as to why FCFS was selected? Would the results change in case PS were implemented instead?}
  \begin{itemize}
  \item
    The results would change if PS tiebreaking was used instead,
    and there is no guarantee that \mserpt{} would still have
    constant approximation ratio with PS tiebreaking.
    Our motivation for using FCFS is two-fold.
    First, regions where \mserpt{}'s rank is constant
    correspond to SERPT valleys,
    in which a job can have better SERPT rank than at the start of the valley,
    so it is natural to choose a nonpreemptive tiebreaking rule for \mserpt{}.
    Second, analyzing SOAP policies with PS tiebreaking is an open problem
    \citep{soap_scully}.
  \end{itemize}
\item
  \textit{As explained in Section 7.2, M-SERPT belongs to the so-called multilevel processor sharing (MLPS) policies introduced and analyzed by Kleinrock. In the case of no arrivals, Gittins' index policy also belongs to MLPS set of policies, ref. \citep{mlps_gittins_aalto}. Thus, in the case of no-arrivals, the difference between Gittins' and M-SERPT will be on the location of thresholds, and possibly on the policy implemented within an interval. In the case of arrivals, Gittins' is not an MLPS policy, as jobs with positive attained service might have higher priority than new jobs.}
  \begin{itemize}
  \item
    We have added a footnote discussing this in \cref{sub:mlps}.
    Note that essentially any SOAP policy,
    specifically any with a piecewise monotonic rank function,
    becomes an MLPS policy with no arrivals.
    We discuss this in more detail below.
  \end{itemize}
\item
  \textit{The performance of Gittins index is bounded by a general result that might be very loose, since it covers size aware policies such as SRPT.}
  \begin{itemize}
  \item
    We address this in \cref{sub:tightening_lower}.
  \end{itemize}
\item
  \textit{The connection between Gittins and MLPS is worth exploring and it could perhaps help obtaining tighter bounds.}
  \begin{itemize}
  \item
    One can think of our proof as using a connection
    between MLPS policies and SOAP policies,
    of which the Gittins-MLPS connection is a special case.
    The more general SOAP-MLPS connection is as follows.
    The class of SOAP policies with monotonic rank functions
    is the same as the class of MLPS policies that use only FCFS and FB levels.
    In particular, any SOAP policy with general rank function~$\rank{}$
    has a related MLPS policy with monotonic rank function~$\rankup{}$
    (\cref{def:increasing_envelope}).
    Thus, one can think of
    \cref{prop:waiting_hills_valleys, prop:residence_hills_valleys} as relating
    the mean waiting and residence times of an arbitrary SOAP policy
    in terms of the corresponding quantities for its related MLPS policy.
  \item
    Our proof currently uses the SOAP-MLPS connection for waiting time
    by way of \cref{prop:waiting_hills_valleys},
    which gives $\E{\waiting[\mserpt]{}}$ exactly
    and a useful lower bound on $\E{\waiting[\gittins]{}}$.
    However, we have not been able to exploit the connection for residence time
    because the bound in \cref{prop:waiting_hills_valleys}
    goes the other direction.
  \item
    We would be open to adding an appendix with
    a more detailed version of the discussion above.
  \end{itemize}
\item
  \textit{In the case of no-arrivals, it would be interesting to characterize the approximation ratio between Gittins and M-SERPT.}
  \begin{itemize}
  \item
    It is possible that a tighter bound is attainable
    in the arrival-free setting.
    However, such a bound would have to use techniques
    fundamentally different from ours.
    This is because the SOAP analysis \citep{soap_scully},
    on which our work is based,
    applies only to the M/G/1 with arrivals.
    A first step might be an arrival-free variation of the SOAP analysis.
  \end{itemize}
\item
  \textit{In the example of \cref{fig:mserpt}, it would be nice to have represented the Gittins' index. This would be useful in understanding how Gittins would operate.}
  \begin{itemize}
  \item
    The purpose of the figure is to illustrate what \mserpt{}
    rather than to compare it to Gittins or SERPT,
    so we prefer leaving the figure uncluttered.
  \end{itemize}
\item
  \textit{\Cref{lem:hill_subset} seems to be the same result as Prop 7 in \citep{mlps_gittins_aalto}.}
  \begin{itemize}
  \item
    The results are not the same, though they use some similar ideas.
    We have added a note clarifying this point
    before the statement of \cref{lem:hill_subset}.
    The difference between the results is that Gittins hill ages
    are not the same as ages at which $\rank[\gittins]{}$ is increasing.
    \Cref{lem:hill_subset} deals with the former,
    while \citet[Proposition~7]{mlps_gittins_aalto} deals with the latter.
  \end{itemize}
\item
  \textit{\Cref{sec:pathological}. For the particular distribution of this section, it would be interesting to depict the rank of M-SERPT and Gittins index, to better understand how both policies differ.}
  \begin{itemize}
  \item
    We chose to explain the key difference in words.
    Specifically,
    the second equation in \cref{sec:pathological}
    shows the key difference between
    $\rank[\mserpt]{}$ and $\rank[\gittins]{}$,
    and the text immediately following it explains
    why this difference causes \mserpt{} to make bad scheduling decisions.
  \end{itemize}
\end{itemize}

\subsection{Reviewer~B}

\begin{itemize}
\item
  \textit{The results, though new, appear a bit limited compared to those in \citep{soap_scully}.}
  \begin{itemize}
  \item
    We assume this is referring to
    \cref{prop:waiting_hills_valleys, prop:residence_hills_valleys},
    which in the previous version were presented as a
    ``new analysis technique''.
    We now clarify that
    we are not replacing the SOAP analysis of \citet{soap_scully}
    but rather presenting a new simplification of it
    (\cref{sub:hills_valleys_outline}).
  \item
    We would like to emphasize that the main result of our work,
    namely \cref{thm:approximation_ratio},
    is absolutely not subsumed by the results of \citet{soap_scully}.
    We address this in \cref{sub:challenges, sub:why_not_soap}.
  \end{itemize}
\item
  \textit{The reviewer wonders how M-SERPT compares with SERPT? M-SERPT has a higher complexity for computing the rank function. It seems fair to expect that it has a better performance. Would it be possible to demonstrate this, either analytically or numerically? \Cref{fig:response_comparison} also does not have the comparison.}
  \begin{itemize}
  \item
    The analytical question remains open.
    We have some preliminary numerical results,
    but have omitted them for lack of space.
    Roughly speaking, SERPT seems never to be
    more than a few percentage points better than \mserpt{},
    including the distribution in \cref{fig:response_comparison},
    but there are a few corner cases where
    SERPT is noticeably worse than \mserpt{}.
    It is possible that the reverse corner cases also exist,
    but we have not yet found them.
  \item
    We omit SERPT from \cref{fig:response_comparison} to reduce clutter.
    The curve is almost on top of the Gittins baseline,
    which makes \mserpt{} at most $4\%$ worse.
  \end{itemize}
\item
  \textit{Given that the scheme M-SERPT depends on distribution information of the job size, it would be useful if the authors elaborate on how sensitive the scheme is to estimation error and how practical it would be for implementation. Schemes like FCFS or PS, though not having strong theoretical guarantees, are immune to such issues. If it is hard to do so analytically, numerical results would help.}
  \begin{itemize}
  \item
    We think that measuring the sensitivity of Gittins, SERPT, and \mserpt{}
    to errors in the specification of the job size distribution~$X$
    is an important open problem.
    It is a significant departure from the goals of this paper,
    so we leave it for future work.
  \end{itemize}
\item
  \textit{Also, an approximation factor 3--5 does not appear to be ``near-optimal'' to the reviewer.}
  \begin{itemize}
  \item
    We still use this phrase because we worried that changing the title
    might be confusing for a one-shot revision,
    but we are open to changing the exact phrase.
  \end{itemize}
\end{itemize}

\subsection{Reviewer~C}

The reviewer had several complaints about the clarity of the writing.
Broadly speaking, these are addressed by improvements we made
in response to point (3) of \cref{sub:main_points_one-shot}.
Here we point out a few of the specific things we have clarified.
\begin{itemize}
\item
  We now clarify that we are not replacing the SOAP analysis
  but rather providing a new simplification of it
  (\cref{sub:hills_valleys_outline}).
  We thus no longer use the phrase ``new analysis technique''.
\item
  Hills and valleys are now defined twice:
  first in \cref{sub:hills_valleys_outline}
  emphasizing intuition and illustrating with a figure,
  then again in \cref{sub:hills_valleys_defs},
  which is more rigorous and covers all the corner cases.
  The reviewer's feedback helped us design this new presentation.
\item
  We now emphasize multiple times
  that \cref{lem:waiting, lem:residence} are
  the core technical steps of our main result.
  To aid the reader in following their complicated proofs,
  we have added new high-level outlines of their respective proof strategies
  in \cref{sub:waiting_outline, sub:residence_outline}.
\item
  We thank the reviewer for their feedback about
  our discussion of the Pessimism Principle.
  We ultimately decided to remove it,
  which helped create space for the high-level outlines mentioned above.
\end{itemize}

The reviewer had one additional comment that does not fit under point~(3).
\begin{itemize}
\item
  \textit{The approach in the paper seems related to older ideas in MDPs of approximate indexability - the authors should in particular take a look at Weber's prevailing charge approach for proving the optimality of the Gittins index policy, and also more recent work on using similar ideas for approximate index policies (for example, the work of Kleinberg, Waggoner and Weyl).}
  \begin{itemize}
  \item
    We are familiar with Weber's prevailing charge argument
    and related work by Kleinberg et al., Dumitriu et al., Singla et al.,
    and others.
    We have not yet been able to find a direct connection.
    It is possible that there is a relationship between this
    and the Gittins-MLPS connection mentioned by Reviewer~A
    (\cref{sub:reviewer_a}).
  \end{itemize}
\end{itemize}

\subsection{Reviewer~D}

The reviewer says
they believe SIGMETRICS is not the appropriate venue for our work.
We disagree: SIGMETRICS has for more than a decade been a premier venue
for work in queueing theory and scheduling.
Specifically, since 2001,
SIGMETRICS has had at least one full session dedicated to
topics in queueing theory,
with over half featuring a full session specifically on scheduling.

\end{anononly}

\end{document}

%% file: fig_mserpt.tex
\begin{tikzpicture}[figure]
  \axes{16.5}{10}{age~$a$}{rank~$\rank{a}$}

  \draw[primary]
  \snake{(0, 2)}{(2.5, 4)}
  -- \snake{(2.5, 4)}{(3.5, 1.5)}
  -- \snake{(3.5, 1.5)}{(4.25, 3.5)}
  -- \snake{(4.25, 3.5)}{(5, 3)}
  -- \snake{(5, 3)}{(8, 5)}
  -- \snake{(10.5, 5)}{(12, 7)}
  -- \snake{(12, 7)}{(13.5, 4)}
  -- \snake{(13.5, 4)}{(15.5, 10)}
  -- \snakefirst{(15.5, 10)}{(17.5, 11)};

  \draw[cutoff, dash phase=0.5pt]
  \snake{(0, 2)}{(2.5, 4)}
  -- \snakesecond{(5, 3)}{(8, 5)}
  -- \snake{(10.5, 5)}{(12, 7)}
  -- \snakesecond{(13.5, 4)}{(15.5, 10)}
  -- \snakefirst{(15.5, 10)}{(17.5, 11)};
\end{tikzpicture}\\[-0.5\captionsqueeze]
\begin{tikzpicture}[figure]
  \draw[primary] (0, 0) -- ++(1.9, 0);
  \node[anchor=west] at (2, 0) {SERPT};
  \draw[cutoff] (6.5, 0) -- ++(1.9, 0);
  \node[anchor=west] at (8.5, 0) {\mserpt{}};
\end{tikzpicture}


%% file: fig_hills_valleys.tex
\begin{tikzpicture}[figure]
  \intervalsnake{hill}{(0, 2)}{(2.5, 4)}{hill}
  \intervalsnake{valley}{(2.5, 4)}{(6.5, 4)}{valley}
  \intervalsnakesecond{hill}{(6.5, 4)}{(8, 5)}{hill}
  \intervalsnake{valley}{(8, 5)}{(10.5, 5)}{valley}
  \intervalsnake{hill}{(10.5, 5)}{(12, 7)}{hill}
  \intervalsnake{valley}{(12, 7)}{(14.5, 7)}{valley}
  \intervalspecial{hill}{(14.5, 7)}{(15.5, 10)}{(16.5, 10.5)}{hill}

  \axes{16.5}{10}{age~$a$}{rank~$\rank{a}$}
  \xguide[$\y(x)$]{2.5}{0}
  \xguide[$x$]{4.25}{0}
  \xguide[$\z(x)$]{6.5}{0}

  \draw[primary]
  \snake{(0, 2)}{(2.5, 4)}
  -- \snake{(2.5, 4)}{(3.5, 1.5)}
  -- \snake{(3.5, 1.5)}{(4.25, 3.5)}
  -- \snake{(4.25, 3.5)}{(5, 3)}
  -- \snake{(5, 3)}{(8, 5)}
  -- \snake{(10.5, 5)}{(12, 7)}
  -- \snake{(12, 7)}{(13.5, 4)}
  -- \snake{(13.5, 4)}{(15.5, 10)}
  -- \snakefirst{(15.5, 10)}{(17.5, 11)};
\end{tikzpicture}


%% file: fig_outline.tex
\newcommand{\dx}{0.25}
\newcommand{\dy}{1.25}
\begin{tikzpicture}[figure,
  node distance=\dy and \dx,
  quantity/.style={draw, semithick, fill=white, rectangle, rounded corners=3pt},
  style={>=stealth, rounded corners=3pt}]
  \newcommand{\nodesplit}[2]{%
    \node[draw, semithick, fill=white, circle, inner sep=0.5pt, below=of #2] (#1)
    {$\displaystyle +$};}
  \newcommand{\nodejoin}[2]{%
    \node[draw, semithick, fill=white, circle, inner sep=0.5pt, above=of #2] (#1)
    {$\displaystyle +$};}
  \node[quantity] (A) at (1.375, 0) {$\displaystyle \E{\response[\mserpt]{}}$};
  \nodesplit{ABC}{A}
  \node[quantity, below right=of ABC] (C)
    {$\displaystyle \E{\residence[\mserpt]{}}$};
  \nodesplit{CDE}{C}
  \node[quantity, below left=of CDE] (D)
    {$\displaystyle \E{\waiting[\mserpt]{}}$};
  \node[quantity, below right=of CDE] (E)
    {$\displaystyle\gp*{\frac{1}{\rho}\log\frac{1}{1 - \rho}}\E{X}$};

  \node (Z') at (-3, -35)
    {\vphantom{$\displaystyle
      \min\curlgp*{
        \max\curlgp*{
          \frac{4}{1 + \sqrt{1 - \rho}},
          \frac{1}{\rho}\log\frac{1}{1 - \rho}
        },
        1 + \frac{4}{1 + \sqrt{1 - \rho}}
      }\E{\response[\gittins]{}}$}};
  \node[quantity] (Z) at (0, -35)
    {$\displaystyle
      \min\curlgp*{
        \max\curlgp*{
          \frac{4}{1 + \sqrt{1 - \rho}},
          \frac{1}{\rho}\log\frac{1}{1 - \rho}
        },
        1 + \frac{4}{1 + \sqrt{1 - \rho}}
      }\E{\response[\gittins]{}}$};
  \nodejoin{XYZ}{Z'}
  \node[quantity, above left=of XYZ] (Y)
    {$\displaystyle
      2 \cdot \frac{2}{1 + \sqrt{1 - \rho}}\E{\waiting[\gittins]{}}
      \vphantom{\min\curlgp*{
        \gp*{\frac{1}{\rho}\log\frac{1}{1 - \rho}}\E{\residence[\gittins]{}},
        \E{\response[\gittins]{}}
      }}$};
  \node[quantity, above right=of XYZ] (X)
    {$\displaystyle
      \min\curlgp*{
        \gp*{\frac{1}{\rho}\log\frac{1}{1 - \rho}}\E{\residence[\gittins]{}},
        \E{\response[\gittins]{}}
      }$};
  \node[quantity, above=3*\dy of Y] (W)
    {$\displaystyle 2 \cdot \E{\waiting[\mserpt]{}}$};
  \nodejoin{BDW}{W}

  \node[quantity] (B) at (W |- C)
    {$\displaystyle \E{\waiting[\mserpt]{}}$};
  \node (X') at (E |- X)
    {\vphantom{$\displaystyle
      \min\curlgp*{
        \gp*{\frac{1}{\rho}\log\frac{1}{1 - \rho}}\E{\residence[\gittins]{}},
        \E{\waiting[\gittins]{}} + \E{\residence[\gittins]{}}
      }$}};

  \draw[->] (A) -- (ABC);
  \draw (ABC) -- (B);
  \draw (ABC) -- (C);
  \draw (B) -- (BDW);
  \draw[->] (C) -- (CDE);
  \draw[postaction={decorate, decoration={markings, mark=at position 1 with
      {\coordinate (D') {};}}}]
    (CDE) -- (D);
  \draw (CDE) -- (E);
  \draw (D) -- (BDW);
  \draw[->] (BDW) -- (W);
  \draw[->] (E) -- (X');
  \draw[->] (W) -- (Y);
  \draw (X) -- (XYZ);
  \draw (Y) -- (XYZ);
  \draw[->] (XYZ) -- (Z');

  \node[right=3 * \dx] (CDElabel) at (CDE) {\cref{lem:residence}};
  \node[left] (EXlabel) at ($ (E.south)!0.5!(X'.north) $)
    {\cref{prop:residence_size, prop:response_srpt}};
  \node[right] (WYlabel) at ($ (W.south)!0.5!(Y.north) $)
    {\cref{lem:waiting}};

  \begin{pgfonlayer}{background}
    \fill[black!12] ($ (D'.north) + (-4.5pt, 0) $)
      rectangle (CDElabel.east |- C.south);
    \fill[black!12] ($ (E.south) + (4.5pt, 0) $)
      rectangle (EXlabel.west |- X.north);
    \fill[black!12] ($ (W.south) + (-4.5pt, 0) $)
      rectangle (WYlabel.east |- Y.north);
  \end{pgfonlayer}

  \iftwocol{%
    \newcommand{\shift}{(0, -50)}}{%
    \newcommand{\shift}{(20.625, -14.8)}}
  \begin{scope}[figure, shift={\shift}]
    \iftwocol{%
      \node[quantity] (B3) at (6, 0) {$B$};
      \coordinate (correction) at (0, 0);
    }{%
      \node[quantity] (B3) at (-2.1875, -12.5) {$B$};
      \coordinate (correction) at (0, 12.5);
    }
    \nodejoin{AB3}{B3}
    \node[quantity, above left=of AB3] (A31) {$A_1$};
    \node[quantity, above right=of AB3] (A32) {$A_2$};
    \draw (A31) -- (AB3);
    \draw (A32) -- (AB3);
    \draw[->] (AB3) -- (B3);
    \node[anchor=north] (label3) at ($ (B3) + (0, -1.2) $) {$A_1 + A_2 \leq B$};

    \coordinate (guide2) at (0, 0);
    \node[quantity] (A2) at ($ (guide2 |- A31) + (correction) $) {$A$};
    \node[quantity, semithick, fill=white, circle, inner sep=0.5pt] (AB2)
      at ($ (A2 |- AB3) + (correction) $) {$\displaystyle +$};
    \node[quantity, below left=of AB2] (B21) {$B_1$};
    \node[quantity, below right=of AB2] (B22) {$B_2$};
    \draw[->] (A2) -- (AB2);
    \draw (AB2) -- (B21);
    \draw (AB2) -- (B22);
    \node[anchor=north] at ($ (guide2) + (0, -1.2) $) {$A \leq B_1 + B_2$};

    \node[quantity] (B1) at (-6 + \dx, 0) {$B$};
    \node[quantity] (A1) at ($ (B1 |- A31) + (correction) $) {$A$};
    \draw[->] (A1) -- (B1);
    \node[anchor=north] at ($ (B1) + (0, -1.2) $) {$A \leq B$};

    \iftwocol{%
      \node[anchor=south] (legend) at ($ (A2) + (0, 1.8) $) {\textsc{Legend}};
      \draw[guide] (legend) ++(-5.4, 1.5) -- ++(10.8, 0);
    }{%
      \node[anchor=south] (legend) at ($ (B22 |- A2)!0.5!(A1) + (0, 1.8) $)
        {\textsc{Legend}};
      \coordinate (divider) at ($ (A1 |- legend) + (-2, 2) $);
      \draw[guide] (divider) -- ($ (divider |- label3) + (0, -2) $);}
  \end{scope}
\end{tikzpicture}


%% file: fig_ratio_bound.tex
\tikzmath{%
  function ratio(\x) {
    if abs(\x) < 0.0001 then {
      return 2;
    } else {
      if abs(1 - \x) < 0.0001 then {
        return 5;
      } else {
        return min(
          max(4 / (1 + sqrt(1 - \x)), 1 / \x * ln(1 / (1 - \x))),
          1 + 4 / (1 + sqrt(1 - \x)));
      };
    };
  };
}
\xscaleload
\begin{tikzpicture}[figure]
  \xguide{8/9}{3}
  \xguide{1}{5}
  \yguide{0}{2}
  \yguide{8/9}{3}
  \yguide{1}{5}

  \axes{1}{5}{$\rho$}{%
    $\dfrac{\E{\response[\mserpt]{}}}{\E{\response[\gittins]{}}}$ bound}

  \draw[secondary, smooth, ultra thick]
    plot [domain=0:0.9587, samples=40] (\x, {ratio(\x)}) --
    plot [domain=0.9587:0.9898, samples=10] (\x, {ratio(\x)}) --
    plot [domain=0.9898:1, samples=10] (\x, {ratio(\x)});
\end{tikzpicture}


%% file: ms.bbl

\begin{thebibliography}{31}


\ifx \showCODEN    \undefined \def \showCODEN     #1{\unskip}     \fi
\ifx \showDOI      \undefined \def \showDOI       #1{#1}\fi
\ifx \showISBNx    \undefined \def \showISBNx     #1{\unskip}     \fi
\ifx \showISBNxiii \undefined \def \showISBNxiii  #1{\unskip}     \fi
\ifx \showISSN     \undefined \def \showISSN      #1{\unskip}     \fi
\ifx \showLCCN     \undefined \def \showLCCN      #1{\unskip}     \fi
\ifx \shownote     \undefined \def \shownote      #1{#1}          \fi
\ifx \showarticletitle \undefined \def \showarticletitle #1{#1}   \fi
\ifx \showURL      \undefined \def \showURL       {\relax}        \fi
\providecommand\bibfield[2]{#2}
\providecommand\bibinfo[2]{#2}
\providecommand\natexlab[1]{#1}
\providecommand\showeprint[2][]{arXiv:#2}

\bibitem[\protect\citeauthoryear{Aalto and Ayesta}{Aalto and Ayesta}{2006a}]%
        {mlps_delay_aalto}
\bibfield{author}{\bibinfo{person}{Samuli Aalto} {and} \bibinfo{person}{Urtzi
  Ayesta}.} \bibinfo{year}{2006}\natexlab{a}.
\newblock \showarticletitle{Mean delay analysis of multi level processor
  sharing disciplines}. In \bibinfo{booktitle}{\emph{INFOCOM 2006. 25th IEEE
  International Conference on Computer Communications. Proceedings}}. IEEE,
  \bibinfo{pages}{1--11}.
\newblock


\bibitem[\protect\citeauthoryear{Aalto and Ayesta}{Aalto and Ayesta}{2006b}]%
        {fb_nonoptimality_aalto}
\bibfield{author}{\bibinfo{person}{S Aalto} {and} \bibinfo{person}{U Ayesta}.}
  \bibinfo{year}{2006}\natexlab{b}.
\newblock \showarticletitle{On the nonoptimality of the foreground-background
  discipline for IMRL service times}.
\newblock \bibinfo{journal}{\emph{Journal of Applied Probability}}
  \bibinfo{volume}{43}, \bibinfo{number}{2} (\bibinfo{year}{2006}),
  \bibinfo{pages}{523--534}.
\newblock


\bibitem[\protect\citeauthoryear{Aalto, Ayesta, Borst, Misra, and
  N{\'u}{\~n}ez-Queija}{Aalto et~al\mbox{.}}{2007}]%
        {ps_beyond_aalto}
\bibfield{author}{\bibinfo{person}{Samuli Aalto}, \bibinfo{person}{Urtzi
  Ayesta}, \bibinfo{person}{Sem Borst}, \bibinfo{person}{Vishal Misra}, {and}
  \bibinfo{person}{Rudesindo N{\'u}{\~n}ez-Queija}.}
  \bibinfo{year}{2007}\natexlab{}.
\newblock \showarticletitle{Beyond processor sharing}. In
  \bibinfo{booktitle}{\emph{ACM SIGMETRICS Performance Evaluation Review}},
  Vol.~\bibinfo{volume}{34}. \bibinfo{publisher}{ACM}, \bibinfo{pages}{36--43}.
\newblock


\bibitem[\protect\citeauthoryear{Aalto, Ayesta, and Nyberg-Oksanen}{Aalto
  et~al\mbox{.}}{2004}]%
        {mlps_two-level_aalto}
\bibfield{author}{\bibinfo{person}{Samuli Aalto}, \bibinfo{person}{Urtzi
  Ayesta}, {and} \bibinfo{person}{Eeva Nyberg-Oksanen}.}
  \bibinfo{year}{2004}\natexlab{}.
\newblock \showarticletitle{Two-level processor-sharing scheduling disciplines:
  mean delay analysis}. In \bibinfo{booktitle}{\emph{ACM SIGMETRICS Performance
  Evaluation Review}}, Vol.~\bibinfo{volume}{32}. ACM,
  \bibinfo{pages}{97--105}.
\newblock


\bibitem[\protect\citeauthoryear{Aalto, Ayesta, and Righter}{Aalto
  et~al\mbox{.}}{2009}]%
        {m/g/1_gittins_aalto}
\bibfield{author}{\bibinfo{person}{Samuli Aalto}, \bibinfo{person}{Urtzi
  Ayesta}, {and} \bibinfo{person}{Rhonda Righter}.}
  \bibinfo{year}{2009}\natexlab{}.
\newblock \showarticletitle{On the {G}ittins index in the {M/G/1} queue}.
\newblock \bibinfo{journal}{\emph{Queueing Systems}} \bibinfo{volume}{63},
  \bibinfo{number}{1} (\bibinfo{year}{2009}), \bibinfo{pages}{437--458}.
\newblock


\bibitem[\protect\citeauthoryear{Aalto, Ayesta, and Righter}{Aalto
  et~al\mbox{.}}{2011}]%
        {mlps_gittins_aalto}
\bibfield{author}{\bibinfo{person}{Samuli Aalto}, \bibinfo{person}{Urtzi
  Ayesta}, {and} \bibinfo{person}{Rhonda Righter}.}
  \bibinfo{year}{2011}\natexlab{}.
\newblock \showarticletitle{Properties of the {G}ittins index with application
  to optimal scheduling}.
\newblock \bibinfo{journal}{\emph{Probability in the Engineering and
  Informational Sciences}} \bibinfo{volume}{25}, \bibinfo{number}{03}
  (\bibinfo{year}{2011}), \bibinfo{pages}{269--288}.
\newblock


\bibitem[\protect\citeauthoryear{Bansal, Kamphorst, and Zwart}{Bansal
  et~al\mbox{.}}{2018}]%
        {rmlf_zwart}
\bibfield{author}{\bibinfo{person}{Nikhil Bansal}, \bibinfo{person}{Bart
  Kamphorst}, {and} \bibinfo{person}{Bert Zwart}.}
  \bibinfo{year}{2018}\natexlab{}.
\newblock \showarticletitle{Achievable performance of blind policies in heavy
  traffic}.
\newblock \bibinfo{journal}{\emph{Mathematics of Operations Research}}
  \bibinfo{volume}{43}, \bibinfo{number}{3} (\bibinfo{year}{2018}),
  \bibinfo{pages}{949--964}.
\newblock


\bibitem[\protect\citeauthoryear{Becchetti and Leonardi}{Becchetti and
  Leonardi}{2004}]%
        {rmlf_leonardi}
\bibfield{author}{\bibinfo{person}{Luca Becchetti} {and}
  \bibinfo{person}{Stefano Leonardi}.} \bibinfo{year}{2004}\natexlab{}.
\newblock \showarticletitle{Nonclairvoyant scheduling to minimize the total
  flow time on single and parallel machines}.
\newblock \bibinfo{journal}{\emph{Journal of the ACM (JACM)}}
  \bibinfo{volume}{51}, \bibinfo{number}{4} (\bibinfo{year}{2004}),
  \bibinfo{pages}{517--539}.
\newblock


\bibitem[\protect\citeauthoryear{Bonald and Proutiere}{Bonald and
  Proutiere}{2002}]%
        {ps_insensitive_bonald}
\bibfield{author}{\bibinfo{person}{Thomas Bonald} {and}
  \bibinfo{person}{Alexandre Proutiere}.} \bibinfo{year}{2002}\natexlab{}.
\newblock \showarticletitle{Insensitivity in processor-sharing networks}.
\newblock \bibinfo{journal}{\emph{Performance Evaluation}}
  \bibinfo{volume}{49}, \bibinfo{number}{1-4} (\bibinfo{year}{2002}),
  \bibinfo{pages}{193--209}.
\newblock


\bibitem[\protect\citeauthoryear{Chakravorty and Mahajan}{Chakravorty and
  Mahajan}{2014}]%
        {gittins_index_computation_chakravorty}
\bibfield{author}{\bibinfo{person}{Jhelum Chakravorty} {and}
  \bibinfo{person}{Aditya Mahajan}.} \bibinfo{year}{2014}\natexlab{}.
\newblock \showarticletitle{Multi-armed bandits, Gittins index, and its
  calculation}.
\newblock \bibinfo{journal}{\emph{Methods and applications of statistics in
  clinical trials: Planning, analysis, and inferential methods}}
  \bibinfo{volume}{2} (\bibinfo{year}{2014}), \bibinfo{pages}{416--435}.
\newblock


\bibitem[\protect\citeauthoryear{Cheung, van~den Berg, and Boucherie}{Cheung
  et~al\mbox{.}}{2006}]%
        {ps_insensitive_cheung}
\bibfield{author}{\bibinfo{person}{Sing-Kong Cheung}, \bibinfo{person}{Hans
  van~den Berg}, {and} \bibinfo{person}{Richard~J Boucherie}.}
  \bibinfo{year}{2006}\natexlab{}.
\newblock \showarticletitle{Insensitive bounds for the moments of the sojourn
  time distribution in the M/G/1 processor-sharing queue}.
\newblock \bibinfo{journal}{\emph{Queueing systems}} \bibinfo{volume}{53},
  \bibinfo{number}{1-2} (\bibinfo{year}{2006}), \bibinfo{pages}{7--18}.
\newblock


\bibitem[\protect\citeauthoryear{Feng and Misra}{Feng and Misra}{2003}]%
        {fb_optimality_misra}
\bibfield{author}{\bibinfo{person}{Hanhua Feng} {and} \bibinfo{person}{Vishal
  Misra}.} \bibinfo{year}{2003}\natexlab{}.
\newblock \showarticletitle{Mixed scheduling disciplines for network flows}. In
  \bibinfo{booktitle}{\emph{ACM SIGMETRICS Performance Evaluation Review}},
  Vol.~\bibinfo{volume}{31}. \bibinfo{publisher}{ACM}, \bibinfo{pages}{36--39}.
\newblock


\bibitem[\protect\citeauthoryear{Gittins, Glazebrook, and Weber}{Gittins
  et~al\mbox{.}}{2011}]%
        {book_gittins}
\bibfield{author}{\bibinfo{person}{John~C. Gittins}, \bibinfo{person}{Kevin~D.
  Glazebrook}, {and} \bibinfo{person}{Richard Weber}.}
  \bibinfo{year}{2011}\natexlab{}.
\newblock \bibinfo{booktitle}{\emph{Multi-armed Bandit Allocation Indices}}.
\newblock \bibinfo{publisher}{John Wiley \& Sons}.
\newblock


\bibitem[\protect\citeauthoryear{Guo and Matta}{Guo and Matta}{2002}]%
        {mlps_analysis_guo}
\bibfield{author}{\bibinfo{person}{Liang Guo} {and} \bibinfo{person}{Ibrahim
  Matta}.} \bibinfo{year}{2002}\natexlab{}.
\newblock \showarticletitle{Scheduling flows with unknown sizes: Approximate
  analysis}. In \bibinfo{booktitle}{\emph{ACM SIGMETRICS Performance Evaluation
  Review}}, Vol.~\bibinfo{volume}{30}. ACM, \bibinfo{pages}{276--277}.
\newblock


\bibitem[\protect\citeauthoryear{Harchol-Balter}{Harchol-Balter}{2013}]%
        {book_harchol-balter}
\bibfield{author}{\bibinfo{person}{Mor Harchol-Balter}.}
  \bibinfo{year}{2013}\natexlab{}.
\newblock \bibinfo{booktitle}{\emph{Performance Modeling and Design of Computer
  Systems: Queueing Theory in Action} (\bibinfo{edition}{1st} ed.)}.
\newblock \bibinfo{publisher}{Cambridge University Press},
  \bibinfo{address}{New York, NY, USA}.
\newblock
\showISBNx{1107027500, 9781107027503}


\bibitem[\protect\citeauthoryear{Kalyanasundaram and Pruhs}{Kalyanasundaram and
  Pruhs}{1997}]%
        {rmlf_pruhs}
\bibfield{author}{\bibinfo{person}{Bala Kalyanasundaram} {and}
  \bibinfo{person}{Kirk~R Pruhs}.} \bibinfo{year}{1997}\natexlab{}.
\newblock \showarticletitle{Minimizing flow time nonclairvoyantly}. In
  \bibinfo{booktitle}{\emph{Proceedings 38th Annual Symposium on Foundations of
  Computer Science}}. IEEE, \bibinfo{pages}{345--352}.
\newblock


\bibitem[\protect\citeauthoryear{Kamphorst and Zwart}{Kamphorst and
  Zwart}{2017}]%
        {fb_heavy_zwart}
\bibfield{author}{\bibinfo{person}{Bart Kamphorst} {and} \bibinfo{person}{Bert
  Zwart}.} \bibinfo{year}{2017}\natexlab{}.
\newblock \showarticletitle{Heavy-traffic analysis of sojourn time under the
  foreground-background scheduling policy}.
\newblock \bibinfo{journal}{\emph{arXiv preprint arXiv:1712.03853}}
  (\bibinfo{year}{2017}).
\newblock


\bibitem[\protect\citeauthoryear{Kleinrock}{Kleinrock}{1967}]%
        {ps_analysis_kleinrock}
\bibfield{author}{\bibinfo{person}{Leonard Kleinrock}.}
  \bibinfo{year}{1967}\natexlab{}.
\newblock \showarticletitle{Time-shared systems: A theoretical treatment}.
\newblock \bibinfo{journal}{\emph{Journal of the ACM (JACM)}}
  \bibinfo{volume}{14}, \bibinfo{number}{2} (\bibinfo{year}{1967}),
  \bibinfo{pages}{242--261}.
\newblock


\bibitem[\protect\citeauthoryear{Kleinrock}{Kleinrock}{1976}]%
        {book_kleinrock}
\bibfield{author}{\bibinfo{person}{Leonard Kleinrock}.}
  \bibinfo{year}{1976}\natexlab{}.
\newblock \bibinfo{booktitle}{\emph{Queueing Systems, Volume 2: Computer
  Applications}}. Vol.~\bibinfo{volume}{66}.
\newblock \bibinfo{publisher}{Wiley New York}.
\newblock


\bibitem[\protect\citeauthoryear{Kleinrock and Muntz}{Kleinrock and
  Muntz}{1972}]%
        {mlps_analysis_kleinrock}
\bibfield{author}{\bibinfo{person}{Leonard Kleinrock} {and}
  \bibinfo{person}{Richard~R Muntz}.} \bibinfo{year}{1972}\natexlab{}.
\newblock \showarticletitle{Processor sharing queueing models of mixed
  scheduling disciplines for time shared system}.
\newblock \bibinfo{journal}{\emph{Journal of the ACM (JACM)}}
  \bibinfo{volume}{19}, \bibinfo{number}{3} (\bibinfo{year}{1972}),
  \bibinfo{pages}{464--482}.
\newblock


\bibitem[\protect\citeauthoryear{Megow and Vredeveld}{Megow and
  Vredeveld}{2014}]%
        {multiple_processors_megow}
\bibfield{author}{\bibinfo{person}{Nicole Megow} {and} \bibinfo{person}{Tjark
  Vredeveld}.} \bibinfo{year}{2014}\natexlab{}.
\newblock \showarticletitle{A Tight 2-Approximation for Preemptive Stochastic
  Scheduling}.
\newblock \bibinfo{journal}{\emph{Mathematics of Operations Research}}
  \bibinfo{volume}{39}, \bibinfo{number}{4} (\bibinfo{year}{2014}),
  \bibinfo{pages}{1297--1310}.
\newblock


\bibitem[\protect\citeauthoryear{Motwani, Phillips, and Torng}{Motwani
  et~al\mbox{.}}{1994}]%
        {nonclairvoyant_motwani}
\bibfield{author}{\bibinfo{person}{Rajeev Motwani}, \bibinfo{person}{Steven
  Phillips}, {and} \bibinfo{person}{Eric Torng}.}
  \bibinfo{year}{1994}\natexlab{}.
\newblock \showarticletitle{Nonclairvoyant scheduling}.
\newblock \bibinfo{journal}{\emph{Theoretical Computer Science}}
  \bibinfo{volume}{130}, \bibinfo{number}{1} (\bibinfo{year}{1994}),
  \bibinfo{pages}{17--47}.
\newblock


\bibitem[\protect\citeauthoryear{Nair, Wierman, and Zwart}{Nair
  et~al\mbox{.}}{2010}]%
        {lps_tail_zwart}
\bibfield{author}{\bibinfo{person}{Jayakrishnan Nair}, \bibinfo{person}{Adam
  Wierman}, {and} \bibinfo{person}{Bert Zwart}.}
  \bibinfo{year}{2010}\natexlab{}.
\newblock \showarticletitle{Tail-robust scheduling via limited processor
  sharing}.
\newblock \bibinfo{journal}{\emph{Performance Evaluation}}
  \bibinfo{volume}{67}, \bibinfo{number}{11} (\bibinfo{year}{2010}),
  \bibinfo{pages}{978--995}.
\newblock


\bibitem[\protect\citeauthoryear{Raz, Levy, and Avi-Itzhak}{Raz
  et~al\mbox{.}}{2004}]%
        {fairness_raz}
\bibfield{author}{\bibinfo{person}{David Raz}, \bibinfo{person}{Hanoch Levy},
  {and} \bibinfo{person}{Benjamin Avi-Itzhak}.}
  \bibinfo{year}{2004}\natexlab{}.
\newblock \showarticletitle{A resource-allocation queueing fairness measure}.
\newblock \bibinfo{journal}{\emph{ACM SIGMETRICS Performance Evaluation
  Review}} \bibinfo{volume}{32}, \bibinfo{number}{1} (\bibinfo{year}{2004}),
  \bibinfo{pages}{130--141}.
\newblock


\bibitem[\protect\citeauthoryear{Righter and Shanthikumar}{Righter and
  Shanthikumar}{1989}]%
        {dhr_ihr_optimality_righter}
\bibfield{author}{\bibinfo{person}{Rhonda Righter} {and}
  \bibinfo{person}{J~George Shanthikumar}.} \bibinfo{year}{1989}\natexlab{}.
\newblock \showarticletitle{Scheduling multiclass single server queueing
  systems to stochastically maximize the number of successful departures}.
\newblock \bibinfo{journal}{\emph{Probability in the Engineering and
  Informational Sciences}} \bibinfo{volume}{3}, \bibinfo{number}{3}
  (\bibinfo{year}{1989}), \bibinfo{pages}{323--333}.
\newblock


\bibitem[\protect\citeauthoryear{Righter, Shanthikumar, and Yamazaki}{Righter
  et~al\mbox{.}}{1990}]%
        {dhr_dmrl_optimality_righter}
\bibfield{author}{\bibinfo{person}{Rhonda Righter}, \bibinfo{person}{J~George
  Shanthikumar}, {and} \bibinfo{person}{Genji Yamazaki}.}
  \bibinfo{year}{1990}\natexlab{}.
\newblock \showarticletitle{On extremal service disciplines in single-stage
  queueing systems}.
\newblock \bibinfo{journal}{\emph{Journal of Applied Probability}}
  \bibinfo{volume}{27}, \bibinfo{number}{2} (\bibinfo{year}{1990}),
  \bibinfo{pages}{409--416}.
\newblock


\bibitem[\protect\citeauthoryear{Schrage}{Schrage}{1968}]%
        {srpt_optimal_schrage}
\bibfield{author}{\bibinfo{person}{Linus Schrage}.}
  \bibinfo{year}{1968}\natexlab{}.
\newblock \showarticletitle{A proof of the optimality of the shortest remaining
  processing time discipline}.
\newblock \bibinfo{journal}{\emph{Operations Research}} \bibinfo{volume}{16},
  \bibinfo{number}{3} (\bibinfo{year}{1968}), \bibinfo{pages}{687--690}.
\newblock


\bibitem[\protect\citeauthoryear{Scully, Harchol-Balter, and
  Scheller-Wolf}{Scully et~al\mbox{.}}{2018}]%
        {soap_scully}
\bibfield{author}{\bibinfo{person}{Ziv Scully}, \bibinfo{person}{Mor
  Harchol-Balter}, {and} \bibinfo{person}{Alan Scheller-Wolf}.}
  \bibinfo{year}{2018}\natexlab{}.
\newblock \showarticletitle{SOAP: One Clean Analysis of All Age-Based
  Scheduling Policies}.
\newblock \bibinfo{journal}{\emph{Proc. ACM Meas. Anal. Comput. Syst.}}
  \bibinfo{volume}{2}, \bibinfo{number}{1}, Article \bibinfo{articleno}{16}
  (\bibinfo{date}{April} \bibinfo{year}{2018}), \bibinfo{numpages}{30}~pages.
\newblock
\showISSN{2476-1249}
\urldef\tempurl%
\url{https://doi.org/10.1145/3179419}
\showDOI{\tempurl}


\bibitem[\protect\citeauthoryear{Wierman}{Wierman}{2007}]%
        {fairness_wierman}
\bibfield{author}{\bibinfo{person}{Adam Wierman}.}
  \bibinfo{year}{2007}\natexlab{}.
\newblock \showarticletitle{Fairness and classifications}. In
  \bibinfo{booktitle}{\emph{ACM SIGMETRICS Performance Evaluation Review}},
  Vol.~\bibinfo{volume}{34}. \bibinfo{publisher}{ACM}, \bibinfo{pages}{4--12}.
\newblock


\bibitem[\protect\citeauthoryear{Wierman, Harchol-Balter, and Osogami}{Wierman
  et~al\mbox{.}}{2005}]%
        {smart_insensitive_wierman}
\bibfield{author}{\bibinfo{person}{Adam Wierman}, \bibinfo{person}{Mor
  Harchol-Balter}, {and} \bibinfo{person}{Takayuki Osogami}.}
  \bibinfo{year}{2005}\natexlab{}.
\newblock \showarticletitle{Nearly insensitive bounds on {SMART} scheduling}.
  In \bibinfo{booktitle}{\emph{ACM SIGMETRICS Performance Evaluation Review}},
  Vol.~\bibinfo{volume}{33}. ACM, \bibinfo{pages}{205--216}.
\newblock


\bibitem[\protect\citeauthoryear{Yamazaki and Sakasegawa}{Yamazaki and
  Sakasegawa}{1987}]%
        {lps_design_yamazaki}
\bibfield{author}{\bibinfo{person}{Genji Yamazaki} {and}
  \bibinfo{person}{Hirotaka Sakasegawa}.} \bibinfo{year}{1987}\natexlab{}.
\newblock \showarticletitle{An optimal design problem for limited processor
  sharing systems}.
\newblock \bibinfo{journal}{\emph{Management Science}} \bibinfo{volume}{33},
  \bibinfo{number}{8} (\bibinfo{year}{1987}), \bibinfo{pages}{1010--1019}.
\newblock


\end{thebibliography}
